%% file: main.tex
\documentclass[11pt]{article}
\usepackage[utf8]{inputenc}
\usepackage{amsmath,amsfonts}
\usepackage{gensymb}
\usepackage{amssymb}
\usepackage[left=1in, right=1in, top=1in, bottom=1in]{geometry}
\usepackage{graphicx}
\usepackage{tikz}
\usetikzlibrary{positioning}
\usepackage{mathtools}
\usepackage{amsthm}
\usepackage[numbers]{natbib}
\usetikzlibrary{arrows,automata,calc}
\usepackage{nccmath}
\usepackage{xspace}
\usepackage{xcolor}
\usepackage{array}
\definecolor{MyBlue}{rgb}{0.12, 0.12, 0.76}
\usepackage[colorlinks,allcolors=MyBlue]{hyperref}
\usepackage{pifont}
\newcommand{\cmark}{\ding{51}}
\newcommand{\xmark}{\ding{55}}

\makeatletter
\newcommand{\thickhline}{%
    \noalign {\ifnum 0=`}\fi \hrule height 1.4pt
    \futurelet \reserved@a \@xhline
}
\newcolumntype{"}{@{\hskip\tabcolsep\vrule width 1.4pt\hskip\tabcolsep}}
\makeatother

\usepackage{thmtools}
\usepackage{thm-restate}

\newtheorem{theorem}{Theorem}[section]
\newtheorem{lemma}{Lemma}[section]

\newtheorem{definition}{Definition}[section]

\usepackage{algorithm}
\usepackage{algpseudocode}
\usepackage{algorithmicx}
\let\oldReturn\Return
\renewcommand{\Return}{\State\oldReturn}
\algtext*{EndWhile}
\algtext*{EndIf}
\algtext*{EndForAll}
\algtext*{EndFor}
\algtext*{EndFunction}

\makeatletter 
\newenvironment{mechanism}[1][htb]{%
    \renewcommand{\ALG@name}{Mechanism}
   \begin{algorithm}[#1]%
  }{\end{algorithm}}
\makeatother

\DeclareMathOperator*{\argmax}{arg\,max}

\allowdisplaybreaks 

\newcommand\bbr{\mathbb{R}}
\newcommand\bbrpos{\mathbb{R}_{\ge 0}}
\newcommand\ep{\varepsilon}
\newcommand{\x}{\mathbf{x}}
\newcommand{\y}{\mathbf{y}}

\newcommand{\bfu}{\mathbf{u}}

\newcommand{\f}{\mathbf{f}}
\newcommand{\g}{\mathbf{g}}

\newcommand{\B}{\mathbf{b}}
\newcommand{\barn}{\bar{N}}
\newcommand{\C}{C}
\newcommand{\calC}{\mathcal{C}}

\newcommand\xprime{\mathbf{x'}}

\newcommand\bprime{\mathbf{b'}}
\newcommand\fprime{\mathbf{f'}}
\newcommand\gprime{\mathbf{g'}}

\newcommand\atp{\mathcal{ATP}}

\newcommand\xs{\mathbf{x^*}}

\newcommand\tilr{\tilde{R}}

\newcommand\todo[1]{}

\begin{document}

\title{Optimal Nash Equilibria for Bandwidth Allocation}
\author{Benjamin Plaut}
\date{bplaut@cs.stanford.edu\\ Stanford University}

\maketitle

\input{abstract}

\input{intro}

\input{related_work}

\input{results}

\input{model}

\input{reduction_intro}

\input{tp_to_pc}

\input{pc_to_tp}

\input{scaling}

\input{ces}

\input{veto}

\input{dse_intro}

\input{maxmin_dse}

\input{maxmin_nash}

\input{conclusion}

\section*{Acknowledgements}

This work would not have been possible without my advisor Ashish Goel, who I would like to thank for continued guidance and feedback. This research was supported by NSF Graduate Research Fellowship under grant DGE-1656518.

\bibliographystyle{plain}
\bibliography{refs}

\end{document}

%% file: abstract.tex
\begin{abstract}
In bandwidth allocation, competing agents wish to transmit data along paths of links in a network, and each agent's utility is equal to the minimum bandwidth she receives among all links in her desired path. Recent market mechanisms for this problem have either focused on only Nash welfare~\cite{branzei_nash_2017}, or ignored strategic behavior~\cite{goel_beyond_2018}. We propose a nonlinear variant of the classic trading post mechanism, and show that for almost the entire family of CES welfare functions (which includes maxmin welfare, Nash welfare, and utilitarian welfare), every Nash equilibrium of our mechanism is optimal. We also prove that fully strategyproof mechanisms for this problem are impossible in general, with the exception of maxmin welfare. More broadly, our work shows that even small modifications (such as allowing nonlinear constraints) can dramatically increase the power of market mechanisms like trading post.
\end{abstract}

%% file: intro.tex
\section{Introduction}\label{sec:intro}

Bandwidth allocation is a classic resource allocation problem where competing agents wish to transmit data across paths in a network. Each link has a fixed capacity, and each agent's utility is equal to the minimum bandwidth she receives among all links in her desired path, i.e., the rate at which she is able to transmit data. We follow the standard model of Kelly et al.~\cite{kelly_1998_rate}, where there are no monetary payments, and each agent's path is fixed in advance. 

Although one could consider a model where bandwidth allocation and routing are handled simultaneously (i.e., by allowing agents to choose their paths), that would be less accurate in terms of how the internet actually works: routing (which is handled by IP) and bandwidth allocation (which is handled by TCP) are generally separate problems\footnote{See Section~\ref{sec:related} for a discussion of routing games.}. This paper is about bandwidth allocation, where pricing-based schemes (like trading post) naturally correspond to signaling mechanisms that indicate which links are congested, and an end-point protocol like TCP~\cite{cerf_1974_protocol} can be thought of as agent responses. One of the foundational works in the area of bandwidth allocation is Kelly et al.~\cite{kelly_1998_rate}, whose pricing scheme results in the allocation maximizing Nash welfare (the product of utilities).

In this paper, we take the role of a social planner,  whose goal is to design a mechanism that leads to a ``desirable" outcome (for some definition of ``desirable"). We study this through the lens of \emph{implementation theory}. A mechanism is said to Nash-implement a social choice rule $\Psi$ (for example, $\Psi$ could denote Nash welfare maximization) if every problem instance has least one Nash equilibrium, and every Nash equilibrium outcome is optimal with respect to $\Psi$. This is similar to saying that the price of anarchy -- the ratio of the optimum and the ``worst" Nash equilibrium -- of the mechanism is 1.\footnote{The price of anarchy~\cite{papa_worst_1999} concept applies only when $\Psi$ can be written as the maximization of some cardinal function. This is true when $\Psi$ denotes Nash welfare maximization, but is not true in general.} In this paper, we focus on pure Nash equilibria, i.e., we do not consider randomized strategies.

The result of Kelly et al.~\cite{kelly_1998_rate} assumes that agents are not strategic, and thus the Nash equilibria of their mechanism may be poor. In contrast, our augmented trading post mechanism will lead to optimal Nash equilibria, not just for Nash welfare, but for an entire family of welfare functions.

\subsection{Trading post}\label{sec:tp-intro}

Our main tool will be an augmented version of the \emph{trading post} mechanism. In the standard trading post mechanism, each agent $i$ submits a bid $b_{ij} \in \bbrpos$ on each good $j$, with the constraint that $\sum_j b_{ij} \le 1$ for each agent $i$. Let $x_{ij}$ be the fraction of good $j$ that agent $i$ receives: then trading post's allocation rule is $x_{ij} = \frac{b_{ij}}{\sum_k b_{kj}}$. In words, each agent receives a share of the good proportional to her share of the aggregate bid on that good. The bids consist of ``fake money": agents have no value for leftover money.

Trading post has the desirable property that the information requirements are quite light. Each agent's best response only depends on the aggregate bid of the other agents (i.e., $\sum_{k \ne i} b_{kj}$), not on their individual bids. Furthermore, the allocation rule is decentralized in the sense that there is no centralized price computation, and each link $j$ only needs to know the bids $b_{1j}, b_{2j}, \dots b_{nj}$.


However, the vanilla version of trading post also has limitations. First of all, it is not even guaranteed to have a Nash equilibrium for every problem instance.\footnote{This happens when there is a good that has large enough supply that is not the ``rate limiting factor" for any agent; see Sections~\ref{sec:related} and \ref{sec:price-zero} for additional discussion.} A partial solution to this was proposed by~\cite{branzei_nash_2017}. For every $\ep > 0$, they gave a modified version of trading post (parameterized by $\ep$) that always has a Nash equilibrium, and where every Nash equilibrium attains at least $1-\ep$ of the maximum possible Nash welfare.\footnote{They study \emph{Leontief utilities}, which is a generalization of bandwidth allocation to the setting where agents may desire goods in different proportions.} In the language of implementation theory, this mechanism Nash-implements a $1-\ep$ approximation of Nash welfare. In the course of our main result, we will strengthen this to full Nash implementation. It is important to note that their mechanism still uses the linear constraint of $\sum_j b_{ij} \le 1$; their modification has to do with a minimum allowable bid (see Section~\ref{sec:related} for additional discussion).

In this paper, we augment the trading post mechanism by allowing nonlinear bid constraints: instead of $\sum_j b_{ij} \le 1$, we require $\sum_j f_j(b_{ij}) \le 1$ for each agent $i$, where each $f_j$ is a nondecreasing function chosen by us ahead of time. Importantly, all agents are still subject to the same bid constraint, and we use the same allocation rule of $x_{ij} = \frac{b_{ij}}{\sum_k b_{kj}}$. This novel augmentation allows us to Nash-implement a wide range welfare functions, as opposed to just Nash welfare. Specifically, we will Nash-implement almost the entire family of CES welfare functions (see Section~\ref{sec:results} for more details). This is our main result.



\subsection{CES welfare functions}\label{sec:ces-intro}

A welfare function~\cite{bergson_1938_reformulation, samuelson_1947_foundations} assigns a real number to each possible outcome, with higher numbers (i.e, higher welfare) indicating outcomes that are more desirable to the social planner. Different welfare functions represent different priorities: in particular the tradeoff of overall efficiency and individual fairness. For any constant $\rho \in (-\infty, 0) \cup (0,1]$, the \emph{constant elasticity of substitution} (CES) welfare function is defined by
\[
\Big(\sum\limits_{i = 1}^n u_i^{\rho}\Big)^{1/\rho}
\]
where $u_i$ is agent $i$'s utility, and $\rho \in \bbr$ is the elasticity parameter. When $\rho = 1$, this is the utilitarian welfare, i.e., sum of utilities. Taking limits as $\rho$ goes to $-\infty$ and 0 yields maxmin welfare (the minimum utility)~\cite{rawls_1971_theory, sen_1976_welfare, sen_1977_social} and Nash welfare (the product of utilities)~\cite{kaneko_nash_1979, nash_bargaining_1950}, respectively. This class of welfare functions was first proposed by Atkinson~\cite{atkinson_1970_measurement} (although under a different name), and further developed by \cite{blackorby_1978_measures}. See \cite{moulin_2003_fair} for a modern introduction to this class of welfare functions.

The closer $\rho$ gets to $-\infty$, the more the social planner cares about individual equality (maxmin welfare being the extreme case of this), and the closer $\rho$ gets to 1, the more the social planner cares about overall societal good (utilitarian welfare being the extreme case of this). The CES welfare function (as opposed to the CES agent utility function) has received almost no attention in the computational economics community, despite being well-studied in the traditional economics literature~\cite{atkinson_1970_measurement, blackorby_1978_measures}.

These welfare functions also admit an axiomatic characterization~\cite{moulin_2003_fair}:
\begin{enumerate}
\item Monotonicity: if one agent's utility increases while all others are unchanged, the welfare function should prefer the new allocation.
\item Symmetry: the welfare function should treat all agents the same.
\item Continuity: the welfare function should be continuous.
\item Independence of common scale: scaling all agent utilities by the same factor should not affect which allocations have better welfare than others.
\item Independence of unconcerned agents: when comparing the welfare of two allocations, the comparison should not depend on agents who have the same utility in both allocations. 
\item The Pigou-Dalton principle: all things being equal, the welfare function should prefer more equitable allocations~\cite{dalton_1920_measurement,pigou_1912_wealth}. 
\end{enumerate}

Ignoring monotonic transformations of the welfare function (which of course do not affect which allocations have better welfare than others), the set of welfare functions satisfying these axioms is exactly the set of CES welfare functions with $\rho \in (-\infty,0)\cup(0,1]$\footnote{Without the Pigou-Dalton principle, $\rho > 1$ is also allowed. This can result in unnatural cases where it is optimal to give one agent everything and the rest none, even when this does not maximize the sum of utilities.}, including Nash welfare~\cite{moulin_2003_fair}.\footnote{This actually does not include maxmin welfare, which obeys weak monotonicity but not strict monotonicity.} This axiomatic characterization shows that we are not just focusing on an arbitrary class of welfare functions: CES welfare functions are arguably the most reasonable welfare functions.

Recently, \cite{goel_beyond_2018} showed that for any CES welfare function, nonlinear pricing can be used to obtain market equilibria with optimal CES welfare. However, their equilibrium notion -- \emph{price curve equilibrium} -- assumes that agents are not strategic.

%% file: related_work.tex
\subsection{Related work}\label{sec:related}

\paragraph{Trading post and market games.}

The trading post mechanism -- first proposed by Shapley and Shubik~\cite{shapley_trade_1977}, and sometimes called the ``Shapley-Shubik game"\footnote{A plethora of other names have been applied to this mechanism as well, including the proportional share mechanism~\cite{feldman_proportional_2009}, the Chinese auction~\cite{matros_chinese_2011}, and the Tullock contest in rent seeking~\cite{buchanan_toward_1980}.} --  is an example of a \emph{strategic market game} (for an overview of strategic market games, see~\cite{giraud_strategic_2003}). The study of markets has a long history
in the economics literature~\cite{arrow_existence_1954, brainard_how_2005,  varian_equity_1974, walras_elements_1874}\footnote{Recently, this topic has garnered significant attention in the computer science community as well (see~\cite{vazirani_2007_combinatorial} for an algorithmic exposition).}, but most of this work assumes that agents are \emph{price-taking}, meaning that they treat the market prices are fixed, and do not behave strategically to affect these prices.\footnote{There is some work treating price-taking market models as strategic games; see e.g., \cite{adsul_nash_2010, branzei_fisher_2014, branzei_nash_2017}.} A market game, however, treats the agents as strategic players who wish to selfishly maximize their own utility. Trading post does not have explicit prices set by a centralized authority: instead, prices arise implicitly from agents' strategic behavior. In particular, $\sum_k b_{kj}$ -- the aggregate bid on good $j$ -- functions as the implicit price of good $j$. Although the trading post mechanism is well-defined for any utility functions, the Nash equilibria are not guaranteed to have many nice properties in general, except in the limit as the number of agents goes to infinity~\cite{dubey_theory_1978} (in this case, the trading post Nash equilibria converge to the price-taking market equilibria).

The paper most relevant to ours is \cite{branzei_nash_2017}, which analyzed the performance of trading post (with a linear bid constraint) with respect to Nash welfare. They showed that for Leontief utilities (which generalize bandwidth allocation), a modified trading post mechanism approximates the Nash welfare arbitrarily well. Specifically, for any $\ep > 0$, they gave a mechanism (parameterized by $\ep$) which achieves a $1-\ep$ Nash welfare approximation: there is at least one Nash equilibrium, and every Nash equilibrium has Nash welfare at least $1-\ep$ times the optimal Nash welfare. Thus the price of anarchy is at most $\frac{1}{1-\ep}$; equivalently, this mechanism Nash-implements a $1-\ep$ approximation of Nash welfare. The reason that they were unable to perfectly implement Nash welfare is because when there is a good with supply much larger than other goods\footnote{Specifically, this occurs when a good has price zero. Having a much larger supply than other goods is sufficient but not necessary for this.}, vanilla trading post may not even have a Nash equilibrium. To fix this, they added a minimum allowable bid, and showed that for any $\ep > 0$, there is a minimum bid that gives them a $1-\ep$ Nash implementation. Instead of having a minimum allowable bid, we will add a special bid $\beta$, which will allow us to strengthen this to full Nash implementation (see Section~\ref{sec:price-zero}).

It is worth noting that~\cite{branzei_nash_2017} also considers a broader class of valuations than Leontief, but for this broader class, only a $1/2$ approximation is achieved. Another recent paper gave a strategyproof mechanism achieving a $1/e\approx .368$ approximation of the optimal Nash welfare~\cite{cole_mechanism_2013}. Their $1/e$ approximation guarantee is weaker than the 1/2 guarantee of \cite{branzei_nash_2017} (and the $1-\ep$ guarantee for Leontief), but strategyproofness is sometimes more desirable that Nash implementation. Unfortunately, strategyproofness in the bandwidth allocation setting is generally impossible (Theorem~\ref{thm:ces-not-sp}).

\paragraph{Price-taking markets.}

The simplest mathematical model of a price-taking market is a \emph{Fisher market}, due to Irving Fisher~\cite{brainard_how_2005}. In a Fisher market, there is a set of goods for sale, and each buyer enters the market with a budget she wishes to spend. Each good has a price, and each buyer purchases her favorite bundle among those that are affordable under her budget constraint. Prices are linear, meaning that the cost of a good is proportional to the quantity purchased, and buyers are assumed to have no value for leftover money, so they will always exhaust their entire budgets. A market equilibrium assigns a price to each good so that the demand exactly equals the supply. For a wide class of agent utilities, including bandwidth allocation utilities, an equilibrium is guaranteed to exist~\cite{arrow_existence_1954}.\footnote{Specifically, an equilibrium is guaranteed to exist as long agent utilities are continuous, quasi-concave, and non-satiated. The full Arrow-Debreu model also allows for agents to enter to market with goods themselves and not only money; the necessary conditions on utilities are slightly more complex in that setting.} The seminal work of Eisenberg and Gale showed that for linear prices and a large class of agent utilities (including bandwidth allocation), the market equilibria correspond exactly to the allocations maximizing Nash welfare~\cite{eisenberg_aggregation_1961, eisenberg_consensus_1959}.\footnote{The conditions for the correspondence between Fisher market equilibria and Nash welfare are slightly stricter than those for market equilibrium existence, but are still quite general. Sufficient criteria were given in \cite{eisenberg_aggregation_1961} and generalized slightly by \cite{jain_eisenberggale_2010}.} Furthermore, the prices are equal to the optimal Lagrange multipliers in the convex program for maximizing Nash welfare (the Eisenberg-Gale convex program).

Recently, \cite{goel_beyond_2018} extended this model to allow nonlinear prices, where the cost of a good may be any nondecreasing function of the quantity purchased. These functions are called \emph{price curves}. They showed that for bandwidth allocation, for any $\rho \in (-\infty, 1)$, there exist price curves that make every maximum CES welfare allocation a market equilibrium. Furthermore, these prices take a natural form: the cost of purchasing $x \in \bbrpos$ of good $j$ is $g_j(x) = q_j x^{1-\rho}$, for some nonnegative constants $q_1 \dots q_m$. Interestingly, for $\rho = 0$ -- which denotes Nash welfare -- this function form reduces to a linear price $q_j$, and we know that linear pricing maximizes Nash welfare. Furthermore, $q_1 \dots q_m$ are the optimal Lagrange multipliers in the convex program for maximizing CES welfare.

Trading post with linear bid constraints ($\sum_j b_{ij} \le 1$) can be thought of as a market game equivalent of the Fisher market model: it implements Nash welfare (\cite{branzei_nash_2017} proved a $1-\ep$ approximation, but we will strengthen this to exact implementation), and the implicit trading post prices (the aggregate bids) are equal to the Fisher market equilibrium prices. Our augmented trading post, with bid constraint $\sum_j f_j(b_{ij}) \le 1$, can be thought of as a market game equivalent of the price curves model. The augmented trading post mechanism we use to implement CES welfare will use $f_j(b) = b^{1-\rho}$ for each good $j$, further strengthening this analogy.



\paragraph{Bandwidth allocation.}

Bandwidth allocation has been studied both with and without monetary payments; we focus on the later setting, following the model of Kelly et al.~\cite{kelly_1998_rate}. Although it has been known that different marking schemes (such as RED and CHOKe~\cite{floyd_1993_random, pan_2000_choke}) and versions of TCP lead to different objective functions (eg. \cite{padhye_1998_modeling}), a market-based understanding was developed only for Nash Welfare, starting with the pioneering work of Kelly et al.~\cite{kelly_1998_rate}. Furthermore, the market scheme of Kelly et al. is in the price-taking setting; the only strategic market analysis of bandwidth allocation that we are aware of is the $1-\ep$ approximation of Nash welfare due to~\cite{branzei_nash_2017}.

\paragraph{Routing games.}

A related topic is that of \emph{routing games}. In a routing game, each agent has a fixed source and destination in the network, but chooses which path she uses to get there. Each agent incurs a cost for each link she travels over, and the cost each agent pays is typically nondecreasing function of the total traffic over that link. Each agent wishes to minimize the total cost she incurs by strategically choosing which path to follow. In the standard bandwidth allocation model, each agent has a fixed path, and her goal is to maximize the total amount of flow she is able to send from her source to her destination (which is equal to the minimum bandwidth she receives among links in her path). Instead of choosing which path to follow, each agent's strategy is how she bids (or more generally, how she interacts with the allocation mechanism). For an overview of routing games, see~\cite{roughgarden_selfish_2005}.

\paragraph{Implementation theory.}

Implementation theory is the study of designing mechanisms whose outcomes coincide with some desirable social choice rule. A social choice rule could be the maximization of a cardinal function, such as a CES welfare function, or something else, such as the set of Pareto optimal allocations. A full survey is outside the scope of this paper; we direct the interested reader to \cite{maskin_2002_book}.

The ``outcome" of a mechanism is not really well-defined; we need to specify a \emph{solution concept}. The solution concept that we focus on for most of this paper is Nash equilibrium. Possibly the most crucial result regarding implementation in Nash equilibrium (Nash implementation, for short) is due to Maskin~\cite{maskin_nash_1999}, who identified a necessary condition for Nash implementation, and a partial converse. He showed that in a very general environment (much broader than bandwidth allocation), any Nash-implementable social choice rule must satisfy what he calls \emph{monotonicity}. Monotonicity, in combination with a property called \emph{no veto power}, is sufficient for Nash implementation. In Section~\ref{sec:no-veto-power}, we show that CES welfare functions do not satisfy no veto power, and so cannot be Nash-implemented by Maskin's approach.

%% file: results.tex
\subsection{Our results}\label{sec:results}

\renewcommand{\arraystretch}{1.6}
\begin{table*}
\centering
\begin{tabular}{ c"c|c|c} 
  & $\rho = -\infty$ & $\rho \in (-\infty, 1)$, & $\rho = 1$\\ 
  \thickhline
Nash-implementable? & \cmark $\ $ (Thm.~\ref{thm:maxmin-nash})
 		& \cmark $\ $ (Thm.~\ref{thm:ces}) 
		& \textbf{?}\\
 \hline
DSE-implementable? & \cmark $\ $(Thm.~\ref{thm:maxmin})
 	 & \xmark $\ $(Thm.~\ref{thm:ces-not-sp}) & \xmark $\ $(Thm.~\ref{thm:ces-not-sp})\\
\end{tabular}
\caption{A summary of our main implementation results. Here $\rho = -\infty$ denotes maxmin welfare, $\rho \in (-\infty, 1)$ includes Nash welfare as $\rho = 0$, and $\rho = 1$ denotes utilitarian welfare. DSE stands for ``dominant strategy equilibrium". ``\cmark" indicates that the type of implementation specified by the row is possible for the social choice rule specified by the column, while ``\xmark" indicates that we give a counterexample, and ``\textbf{?}" indicates an open question.}
\label{tbl:results-table}
\end{table*}
\renewcommand{\arraystretch}{1}

Our results fall into two categories, both summarized by Table~\ref{tbl:results-table}.

\paragraph{Nash-implementing CES welfare functions.}

We view the Nash implementation of CES welfare functions by trading post as our main result (Theorem~\ref{thm:ces}). For each $\rho \in (-\infty, 1)$, we define an augmented trading post mechanism with a nonlinear bid constraint of $\sum_j b_{ij}^{1-\rho} \le 1$ for each agent $i$.\footnote{The reader may notice that for $\rho = 0$ -- which corresponds to Nash welfare -- this constraint reduces to the standard linear constraint of $\sum_j b_{ij} \le 1$, which is what we should expect: we know from~\cite{branzei_nash_2017} that trading post with the linear constraint leads to good Nash welfare.} We denote this mechanism by $\atp(\rho)$. We show that $\atp(\rho)$ has at least one Nash equilibrium, and that all of its Nash equilibria maximize CES welfare. 

Our result improves that of \cite{goel_beyond_2018} by strengthening their price curve equilibrium (which assumes agents are not strategic) to a strategic equilibrium, and improves that of \cite{branzei_nash_2017} by generalizing from just Nash welfare to all CES welfare functions (except $\rho = 1$) and strengthening their $1-\ep$ approximation to exact implementation.\footnote{It is worth noting that the result of \cite{branzei_nash_2017} holds for Leontief utilities, a generalization of bandwidth allocation utilities.} Furthermore, because the price curve equilibria can be computed in polynomial time~\cite{goel_beyond_2018}, our Nash equilibria can also be computed in polynomial time.

Our proof makes use of the following results (stated informally):
\begin{enumerate}
\item Theorem~\ref{thm:tp-to-pc}: Any Nash equilibrium of $\atp$ can be converted into an ``equivalent" price curve equilibrium.
\item Theorem~\ref{thm:pc-to-tp}: Any price curve equilibrium can be converted into an ``equivalent" Nash equilibrium of $\atp$.
\item Lemma~\ref{lem:ces-price-curves}~\cite{goel_beyond_2018}: If $\x$ is a maximum CES welfare allocation, then there exist price curves $\g$ of the form $g_j(x) = q_j x^{1-\rho}$ such that $(\x, \g)$ is a price curve equilibrium.
\item Lemma~\ref{lem:ces-sufficient}~\cite{goel_beyond_2018}: If $(\x, \g)$ is a price curve equilibrium and each $g_j$ has the form $g_j(x) = q_j x^{1-\rho}$, then $\x$ is a maximum CES welfare allocation.
\end{enumerate}
Lemmas~\ref{lem:ces-price-curves} and \ref{lem:ces-sufficient} together imply that $\x$ is a maximum CES welfare allocation if and only if it is a price curve equilibrium with respect to some price curves $\g$ of the form $g_j(x) = q_j x^{1-\rho}$ (where $q_1\dots q_m$ are nonnegative constants). Theorems~\ref{thm:tp-to-pc} and \ref{thm:pc-to-tp} allow us to convert between price curve equilibria and Nash equilibria of $\atp$, and thus enable us to apply Lemmas~\ref{lem:ces-price-curves} and \ref{lem:ces-sufficient} to the Nash equilibria of $\atp(\rho)$. Specifically, Theorem~\ref{thm:tp-to-pc} in combination with Lemma~\ref{lem:ces-sufficient} will show that any Nash equilibrium of $\atp(\rho)$ maximizes CES welfare, and Theorem~\ref{thm:pc-to-tp} in combination with Lemma~\ref{lem:ces-price-curves} will show that $\atp(\rho)$ has at least one Nash equilibrium. 

Section~\ref{sec:reduction} is devoted to proving our reduction between price curve equilibrium and Nash equilibria of trading post: Theorems~\ref{thm:tp-to-pc} and \ref{thm:pc-to-tp}. This reduction is the main tool we use to Nash-implement CES welfare maximization. Section~\ref{sec:ces} then uses this reduction, in combination with Lemmas~\ref{lem:ces-price-curves} and \ref{lem:ces-sufficient}, to prove our main theorem: Theorem~\ref{thm:ces}.

Our trading post approach breaks down for $\rho = -\infty$ and $\rho = 1$. We are able to Nash-implement $\rho = -\infty$ by a different mechanism (see below), but we were not able to resolve whether $\rho = 1$ is Nash-implementable. We leave this as an open question.

\paragraph{Results for dominant strategy implementation and maxmin welfare.}

A natural question is whether these results can be improved from Nash implementation to implementation in dominant strategy equilibrium (DSE). Section~\ref{sec:dse} shows that the answer is mostly no: for any $\rho \in (-\infty, 1]$, there is no mechanism which DSE-implements CES welfare maximization (Theorem~\ref{thm:ces-not-sp}). We do this by showing that there is no strategyproof mechanism for this problem: the revelation principle tells us that DSE-implementability implies strategyproofness, so impossibility of strategyproofness implies impossibility of DSE implementation.

On the positive side, we show that maxmin welfare ($\rho = -\infty$) can in fact be DSE-implemented by a simple revelation mechanism (Theorem~\ref{thm:maxmin}). This is actually stronger than strategyproofness: strategyproofness requires truth-telling to be \emph{a} DSE, but does not rule out the possibility of additional dominant strategy equilibria that are not optimal. In contrast, DSE implementation requires \emph{every} DSE to be optimal.

Although every DSE is also a Nash equilibrium, DSE-implementability does \emph{not} imply Nash-implementability \cite{dasgupta_impl_1979}. A DSE implementation requires every DSE to be optimal, but there could be Nash equilibria (which are not dominant strategy equilibria) that are not optimal. This means that Theorem~\ref{thm:maxmin} does not imply Nash-implementability of maxmin welfare. In fact, our revelation mechanism which DSE-implements maxmin welfare is not a Nash implementation: there exist Nash equilibria which are not optimal (see Section~\ref{sec:maxmin-not-nash} for an example). Our last result of Section~\ref{sec:dse} is that there is a different mechanism which does Nash-implement maxmin welfare (Theorem~\ref{thm:maxmin-nash}).\footnote{The mechanism for Theorem~\ref{thm:maxmin-nash} is unrelated to trading post: our trading post approach breaks down for both maxmin welfare and utilitarian welfare. This is because $g_j(x) = q_j x^{1-\rho}$ is not a valid price curve when $\rho \to -\infty$ or when $\rho = 1$.}

The rest of the paper is structured as follows. Section~\ref{sec:model} formally defines the models of bandwidth allocation, price curves, trading post, and implementation theory. Section~\ref{sec:reduction} presents our reduction between price curves and our augmented trading post mechanism. In Section~\ref{sec:ces}, we use this reduction to Nash-implement CES welfare maximization for $\rho \in (-\infty, 1)$. Finally, Section~\ref{sec:dse} handles DSE-implementation and maxmin welfare.

%% file: model.tex
\section{Model}\label{sec:model}

Let $N = \{1,2,\ldots n\}$ be a set of agents, and let $M = \{1,2,\dots m\}$ be a set of divisible goods, each representing a link in a network. Throughout the paper, we use $i$ and $k$ to refer to agents and $j$ and $\ell$ to refer to goods. Let $s_j$ denote the available supply of good $s_j$. The social planner needs to determine an \emph{allocation} $\x \in \bbr_{\ge 0}^{n\times m}$, where $x_i \in \bbr_{\ge 0}^m$ is the \emph{bundle} of agent $i$, and $x_{ij} \in [0,s_j]$ is the quantity of good $j$ allocated to agent $i$. An allocation cannot allocate more than the available supply: $\x$ is a valid allocation if and only if $\sum_i x_{ij} \leq s_j$ for all $j$.

Agent $i$'s utility for a bundle $x_i$ is denoted by $u_i(x_i)  \in \bbr_{\ge 0}$. We assume agents have \emph{bandwidth allocation utilities}, which take the form:
\[
u_i(x_i) = \min_{j \in R_i} x_{ij}
\]
where $R_i$ is the set of links that agent $i$ requires. We assume that $R_i \ne \emptyset$ for all $i$, i.e., each agent desires at least one good. It will sometimes be useful to define \emph{weights} $w_{ij}$ where $w_{ij} = 1$ if $j \in R_i$, and 0 otherwise.

Just as agents have utilities over the bundles they receive, we can imagine a social planner who wishes to design a mechanism to maximize some societal welfare function $\Phi(\x)$. One can think of $\Phi$ as the social planner's utility function, which takes as input the agent utilities, instead of a bundle of goods. The most well-studied welfare functions are the \emph{maxmin welfare} $\Phi(\x) = \min_{i \in N} u_i(x_i)$, the \emph{Nash welfare} $\Phi(\x) = \big(\prod_{i \in N} u_i \big)^{1/n}$, and the \emph{utilitarian welfare} $\Phi(\x) = \sum_{i \in N} u_i(x_i)$. As discussed in Section~\ref{sec:intro}, these three welfare functions can be generalized by a CES welfare function:
\[
\Phi_\rho(\x) = \Big(\sum\limits_{i \in N} u_i(x_i)^{\rho}\Big)^{1/\rho}
\]
where $\rho$ is a constant in $(-\infty, 0) \cup (0, 1]$. The limits as $\rho \to -\infty$ and $\rho \to 0$ yield maxmin welfare and Nash welfare, respectively. Throughout the paper, we will use $\rho = -\infty$ and $\rho = 0$ to denote maxmin welfare and Nash welfare (e.g., `` This theorem holds for $\rho \in (\infty, 1)$" would include Nash welfare but not maxmin welfare).

For $\rho \ne 1$, this function is strictly concave in $u_1(x_1)\dots u_n(x_n)$, so every optimal allocation $\x$ has the same utility vector.\footnote{There could be multiple optimal allocations, however. For example, consider one agent who desires two goods with supply $s_1$ and $s_2 > s_1$. The agent's optimal utility will be $s_1$, but we can either allocate the rest of the second good anyway, or leave some unallocated; the utility is unaffected.}

\subsection{Price curves}\label{sec:price curves}

Price curves, introduced by \cite{goel_beyond_2018}, generalize the well-studied Fisher market model. In a Fisher market~\cite{brainard_how_2005}, each good is available for sale and each agent enters the market with a fixed budget she wishes to spend. Each good $j$ has a price $p_j$, and the cost of purchasing $x \in \bbrpos$ of good $j$ is $p_j \cdot x$. Price curves allow the cost of a good to be any nondecreasing function of the quantity purchased. When each price curve $g_j$ is defined by $g_j(x) = p_j\cdot x$, this reduces to the Fisher market model. We do not consider strategic behavior in the price curves model; instead, we use this model as a tool for analyzing the Nash equilibria of our augmented trading post.

Like the Fisher market model, the price curves model assumes that agents have no value for leftover money; this will imply that each agent always spends her entire budget. Throughout the paper, we will assume that all agents have the same budget, and normalize all budgets to 1 without loss of generality.  In this paper, we will assume that price curves are either strictly increasing, or identically zero (denoted $g_j \equiv 0$)\footnote{This assumption is not made in \cite{goel_beyond_2018}, but is helpful for our purposes.}. We also assume that price curves are normalized $(g_j(0) = 0$) and continuous.

Informally, a \emph{price curve equilibrium} assigns a price curve $g_j: \bbrpos \to \bbrpos$ to each good $j$ so that the agents' demand equals supply. Formally, for price curves $\g = (g_1\dots g_m)$, the cost of a bundle $x_i$ is
\[
\C_\g(x_i) = \sum_{j \in M} g_j(x_{ij})
\]
and the \emph{demand set} $D_i(\g)$ is the set of agent $i$'s favorite affordable bundles:
\[
D_i(\g) = \argmax\limits_{x_i \in \bbr^m_{\geq 0}:\ C_\g(x_i) \leq 1} u_i(x_i)
\]
If $g_j$ is strictly increasing for all $j \in M$, an agent with bandwidth allocation utility will only purchase goods in her set $R_i$, and will purchase the exact same quantity of each. When $g_j \equiv 0$, any agent can add more of that good at no additional cost: this cannot improve her utility, but it also cannot hurt it. This complication with zero-price goods is discussed more in Section~\ref{sec:price-zero}.

A price curve equilibrium (PCE) $(\x, \g)$ is an allocation $\x$ and price curves $\g$ such that
\begin{enumerate}
\item Each agent receives a bundle in her demand set: $x_i \in D_i(\g)$.
\item The market clears: for all $j \in M$, $\sum_{i \in N} x_{ij} \leq s_j$, and $\sum_{i \in N} x_{ij} = s_j$ whenever $g_j \not\equiv 0$.\footnote{The definition given in \cite{goel_beyond_2018} omits this condition, because that paper is also interested in equilibria where the supply is not exhausted. Such equilibria cannot be optimal for our purposes, so we disallow them in our definition.}
\end{enumerate}
The second condition states that the demand never exceeds the supply, and that any good whose supply is not completely exhausted must have a price of zero. This implies that no agent has utility for the leftover goods: otherwise she would simply buy more at no additional cost. 

\subsection{The trading post mechanism}

In the standard trading post mechanism, each agent $i$ places a \emph{bid} $b_i \in \bbrpos^m$, where $b_{ij} \in \bbrpos$ is the amount $i$ bids on good $j$. Each agent $i$ must obey the constraint $\sum_{j \in M} b_{ij} \le 1$. We use $\B \in \bbrpos^{m\times n}$ to represent the matrix of all bids.

Each agent receives a fraction of the good in proportion to the fraction of the total bid on that good. Formally,
\[
x_{ij} = \frac{b_{ij}}{\sum_{k \in N} b_{kj}} \cdot s_j
\]
As in the Fisher market model, we assume that agents have no value for leftover money. The aggregate bid on good $j$ is $\sum_{k \in N} b_{ij}$, and can be thought of as the ``price" of good $j$: in fact, this analogy will be crucial in our proofs.

We augment the standard trading post mechanism in two ways. The first is necessary in order to ensure the existence of equilibrium when goods have price zero, and the second is to extend this mechanism to implement CES welfare functions beyond Nash welfare.

\subsubsection{Handling goods with price zero}\label{sec:price-zero}

In Fisher markets, it is possible for some goods to have price zero. This occurs when that good is not the ``rate-limiting factor", i.e., there is enough of that good for everyone and the supply constraint is not tight. This is a problem for standard trading post: in order to receive any amount of good $j$, agent $i$ must bid $b_{ij} > 0$. But if the supply constraint is not tight in the Fisher market setting, there will be at least one agent receiving more of the good than they need. Such an agent will decrease their bid so that she is only receiving what she needs. However, this process will continue infinitely, with agents repeatedly decreasing their bids on this good, but never reaching bid 0. 

To handle this, we present the following modified allocation rule. We allow an additional special bid of $\beta$ so that $b_{ij} \in \bbrpos \cup \{\beta\}$. Conceptually, a bid of 0 indicates that the agent actually does not want the good; bidding $\beta$ indicates that the agent desires the good, but is hoping to get it for free, so to speak. We treat $\beta$ as zero in arithmetic, for example, in the constraint $\sum_{j \in M} b_{ij} \le 1$. Similarly, we interpret $b_{ij} > 0$ to mean $b_{ij} \not\in \{0,\beta\}$.

Our modified allocation rule follows this series of steps:
\begin{enumerate}
\item If at least one agent bids a positive (i.e., neither 0 nor $\beta$) amount on good $j$, we follow the standard trading post rule: $x_{ij} = \mfrac{b_{ij}}{\sum_{k \in N} b_{ij}}s_j$.
\item However, if all agents bid 0 either or $\beta$ on good $j$, then we allow each agent to have as much good $j$ as they want. Specifically, for any agent $i$ with $b_{ij} = \beta$, let $\ell_i$ be an arbitrary good with $b_{i\ell_i} > 0$. Then we allocate $x_{ij} = x_{i\ell_i}$. For completeness, if there is no good $\ell$ with $b_{i\ell} > 0$ (although this will never happen at equilibrium), we set $x_{ij} = 0$.  For agents $i$ bidding 0 on good $j$, we set $x_{ij} = 0$.
\item After following the above steps, for any good $\ell$ where $\sum_{i \in N} x_{i\ell} > s_\ell$ (violating the supply constraint), for all $i\in N$ bidding $\beta$ on good $\ell$, we set $x_{ij} = 0$ for all $j\in M$ as a penalty. In words, if so many agents try to get good $j$ for free that the supply constraint is violated, they are all penalized by receiving nothing. Not to worry: this will never happen at equilibrium.
\end{enumerate}
This modification will allow us to simulate a good having price zero.

\todo{Should I put this in pseudocode?}

It is important that we allow separate bids of 0 and $\beta$. Consider a good $j$ where $b_{kj} \in \{0,\beta\}$ for all $k \in N$. Suppose some agent $i$ does not need good $j$, and bidding $\beta$ would cause the supply constraint to be violated and the Step 3 penalty to be invoked. Such an agent can bid 0 on good $j$, which allows her to still spend no money on this good, without the possibility of invoking the Step 3 penalty.

\subsubsection{Allowing nonlinear constraints}

It will turn out that trading post with the standard constraint of $\sum_{j \in M} b_{ij} \le 1$ implements Nash welfare. To implement other CES welfare functions, let $\f = (f_1\dots f_m)$ be nondecreasing functions from $\bbrpos$ to $\bbrpos$. We call $\f$ the \emph{constraint curves}. Like price curves, we assume that each $f_j$ is continuous and normalized. Unlike price curves, we require each $f_j$ to be strictly increasing: $f_j \equiv 0$ is not allowed. Throughout the paper, we will use $\f, \fprime$ to denote constraint curves and $\g, \gprime$ to denote price curves.

We define the mechanism $\atp(\f)$ as follows. Given bids $\B = (b_1\dots b_n) \in \bbrpos^{n\times m}$, $\atp(\f)$ allocates each good $j$ according to the three-step allocation rule described in the previous section. However, each agent's bid constraint is now
\[
\sum_{j \in M} f_j(b_{ij}) \le 1
\]
We can define $C_\f(b_i)$ like we defined $C_\g(x_i)$ for price curves $\g$ and a bundle $x_i$. Specifically, $C_\f(b_i) = \sum_{j \in M} f_j(b_{ij})$. Thus each agent's constraint is $C_\g(x_i) \le 1$ in the price curves model, and is $C_\f(b_i) \le 1$ in the trading post model.

The most natural case will be when $f_1\dots f_m$ are all the same function. In particular, let $\atp(\rho)$ be the mechanism where $f_j(b) = b^{1-\rho}$ for all $j \in M$. In general, we will use $\atp(\f, \B)$ to denote the allocation $\x$ produced by the mechanism $\atp(\f)$ when agents bid $\B \in \bbrpos^{n\times m}$.

\subsection{Implementation theory}\label{sec:impl}

This section covers only the basic concepts of implementation theory; we direct the reader to \cite{maskin_2002_book} for a broad overview of this area.

A \emph{social choice rule} $\Psi$ takes as input a utility profile $\bfu = u_1\dots u_n$ and returns a set of ``optimal" outcomes. In our case, $\Psi$ will represent maximizing a CES welfare function. Define $\Psi_\rho(\bfu)$ by
\[
\Psi_\rho(\bfu) = \argmax_{\x \in \bbr_{\ge 0}^{n\times m}} \Big(\sum_{i \in N} u_i(x_i)^\rho\Big)^{1/\rho}
\]
In general, a social choice rule need not express the maximization of any cardinal function.

Let $\calC$ be a solution concept (e.g., Nash equilibrium), $H$ be a mechanism (sometimes called a ``game form"), and $H(\bfu)$ be the induced game for utility profile $\bfu$.\footnote{In general, the difference between a game and a mechanism is that the game definition includes the agent utilities, whereas a mechanism does not.} Let $\calC(H(\bfu))$ be the set of strategy profiles\footnote{A strategy profile is a list of strategies $S_1\dots S_n$, where $S_i$ is the strategy played by agent $i$. For trading post, a strategy is $b_i \in \bbrpos^m$, and a strategy profile is $\B \in \bbrpos^{m\times n}$.} satisfying $\calC$ for that game. For example, if $\calC$ denotes Nash equilibrium, then $\calC(H(\bfu))$ would be the set of Nash equilibria of the game $H(\bfu)$. To distinguish between equilibrium strategies (e.g., what agents bid) and equilibrium outcomes (e.g., the resulting allocation), we use $\calC_X(H(\bfu))$ to denote the set of outcomes resulting from strategy profiles satisfying $\calC$.

\begin{definition}\label{def:impl}
A mechanism $H$ \emph{implements} a social choice rule $\Psi$ if for any utility profile $\bfu$,
\[
\emptyset \neq \calC_X(H(\bfu)) \subseteq \Psi(\bfu)
\]
\end{definition}

Using the running example of Nash equilibrium, $H$ Nash-implements $\Psi$ if for any utility profile $\bfu$, there is at least one Nash equilibrium, and every Nash equilibrium of $H(\bfu)$ results in an outcome that is optimal under $\Psi$. We denote the set of Nash equilibria of $H(\bfu)$ by $NE(H(\bfu))$, and the set of outcomes resulting from some Nash equilibrium by $NE_X(H(\bfu))$. When only a single utility profile $\bfu$ is under consideration, we will frequently leave $\bfu$ implicit and write $NE(H)$.

It is worth noting that some of the literature refers to Definition~\ref{def:impl} as \emph{weak implementation}, where \emph{full implementation} requires that $\calC_X(H(\bfu)) = \Psi(\bfu)$, i.e., every outcome that is optimal under $\Psi$ should be a Nash equilibrium outcome of $H(\bfu)$. We feel that this distinction is not important in our case, since the utility vector in $\Psi_\rho(\bfu)$ is unique (with the exception of $\rho = 1$, which we do not Nash implement anyway): thus allocations $\x \in \Psi_\rho(\bfu)$ differ only in what they do with leftover supply, i.e., supply that will not affect anyone's utility. If one truly cared about this distinction, our augmented trading post mechanism could be further augmented by allowing each agent another special bid that indicated how much of the leftover supply they wanted. Since these special bids would not affect the utilities, the Nash equilibrium utilities would not be affected, and there would be a combination of leftover supply bids that achieves any maximum CES welfare allocation.\footnote{We would also need to include another penalty step if the leftover supply bids lead to a supply constraint being violated.}

We remind the reader of the following standard definitions:
\begin{enumerate}
\item Nash equilibrium: a strategy profile where no agent can strictly improve her utility by unilaterally changing her strategy. We consider only pure Nash equilibria, i.e., we do not allow randomized strategies.
\item Dominant strategy: a strategy that is optimal regardless of what other agents do.
\item Dominant strategy equilibrium (DSE): a strategy profile where each agent plays a dominant strategy.
\item Strategyproofness: A revelation mechanism (i.e., a mechanism that asks each agent to report her utility function) is \emph{strategyproof} if telling the truth is a dominant strategy for every agent.
\end{enumerate}

DSE-implementability implies strategyproofness via the revelation principle\footnote{See Chapter 9 of \cite{nisan_2007_algorithmic} for an introduction to the revelation principle.}, but it is \emph{not} generally true that any strategyproof social choice rule is DSE-implementable. Strategyproofness ensures that truth-telling is \emph{a} dominant strategy equilibrium, but there could also be bad equilibria that are not consistent with $\Psi$.

By definition, every DSE is also a Nash equilibrium. However, it is \emph{not} generally true that DSE-implementability implies Nash-implementability~\cite{dasgupta_impl_1979}. DSE-implementability requires that every DSE of the mechanism be optimal under $\Psi$, but the mechanism might have additional Nash equilibria (that are not dominant strategy equilibria) that are not consistent with $\Psi$. We will need to take both this and the previous paragraph into account when studying DSE implementation.

We now move on to our results, beginning with our reduction between price curves and $\atp$. This reduction will be the main tool we use to show that $\atp$ Nash-implement CES welfare maximization.

%% file: reduction_intro.tex
\section{Reduction between price curves and augmented trading post}\label{sec:reduction}

In this section, we show that any equilibrium of our augmented trading post mechanism can be transformed into a price curve equilibrium, and vice versa. Section~\ref{sec:ces} will use this result (along with the existence of price curve equilibria maximizing CES welfare, due to \cite{goel_beyond_2018}) to prove that the $\atp(\rho)$ mechanism Nash-implements CES welfare maximization.

Section~\ref{sec:intuition} describes the intuition behind the reduction. Section~\ref{sec:eq-cond} presents some useful necessary and sufficient conditions for price curve equilibrium and trading post Nash equilibrium. Section~\ref{sec:tp-to-pc} shows that any trading post Nash equilibrium can be transformed into a price curve equilibrium (Theorem~\ref{thm:tp-to-pc}), and Section~\ref{sec:pc-to-tp} shows that any price curve equilibrium can be transformed into a trading post Nash equilibrium (Theorem~\ref{thm:pc-to-tp}).

\subsection{Intuition behind the reduction}\label{sec:intuition}

First, notice that augmented trading post and price curves have similar-looking constraints: $\sum_{j \in M} f_j(b_{ij}) \le 1$ and $\sum_{j \in M} g_j(x_{ij}) \le 1$. If $\f = \g$, these constraints become identical, so $b_i$ is a feasible bid if and only if $x_i$ is a feasible purchase subject to price curves $\g$. Suppose that $(\x, \g)$ is a price curve equilibrium. For now, assume each $g_j$ is strictly increasing (the formal proof will also handle the possibility of $g_j\equiv 0$). Let $\xprime$ be the outcome of $\atp(\f)$ when agents bid $\B$ (i.e., $\xprime = \atp(\f, \B)$), and suppose that $b_{ij} = x_{ij}$ for all $i,j$: then
\[
x'_{ij} = \frac{b_{ij}}{\sum_{k \in N} b_{kj}} s_j = \frac{x_{ij}}{\sum_{k \in N} x_{ij}} s_j = x_{ij}
\]
where the last equality uses the fact that $\sum_{k \in N} x_{ij} = s_j$ when $(\x, \g)$ is a PCE and $g_j \not\equiv 0$.

Thus the allocation resulting from $\atp(\f)$ under bids $\B$ is in fact $\x$. Furthermore, since $(\x, \g)$ is a price curve equilibrium, each agent exhausts her price curve constraint: $C_\g(x_i) = 1$. Since $\f = \g$ and $\B = \x$, this implies that $C_\f(b_i) = 1$ for all $i \in N$. Furthermore, in any price curve equilibrium with all nonzero prices, each agent should be spending exclusively on goods in her set $R_i$, and purchasing them in equal amounts. Thus in the trading post outcome $\xprime$, each agent $i$ also also spending exclusively on $j \in R_i$ and acquiring them in equal amounts. 

We claim that $\B$ is a Nash equilibrium of $\atp(\f)$. Suppose the opposite: then there must exist an agent $i$ and an alternate bid $b'_i$ such that bidding $b'_i$ instead of $b_i$ increases her utility. Thus under $b'_i$, she receives strictly more of all goods in $R_i$. But this means that she must be bidding strictly more on each of these goods, which would violate her bid constraint, since $C_\f(b_i) = 1$ is already tight. Therefore $\B$ must be a Nash equilibrium of $\atp(\f)$. 

The above is an informal proof of one direction of the reduction: transforming price curve equilibria into trading post equilibria. Similarly, if we are given a Nash equilibrium $\B$ of $\atp(\f)$, we can let $\g = \f$ (actually, $\g$ will be a scaled version of $\f$) and $\x = \atp(\f, \B)$, and use the same intuition to show that $(\x, \g)$ is a price curve equilibrium.

There are several additional complications. The largest of these is dealing with goods that have price zero in $\g$; indeed, this is the issue that prevents vanilla trading post from implementing Nash welfare maximization~\cite{branzei_nash_2017}. Another difficulty is that in trading post, what you bid depends on others' bids (whereas for price curves, it only depends on $\g$). However, due to the nature of bandwidth allocation utilities, agents will always purchase in proportion to their weights $w_{ij}$, and the outcomes at equilibrium will correspond. We will end up with the following two theorems:

\begin{restatable}{theorem}{thmTpToPc}
\label{thm:tp-to-pc}
Let $\f$ be constraint curves where each $f_j$ is homogenous of degree $\alpha_j$ for some $\alpha_j > 0$. For bids $\B \in NE(\atp(\f))$, define nonnegative constants $a_1\dots a_m$ by $a_j = (\sum_{k \in N} b_{kj}/s_j)^{\alpha_j}$. Define price curves $\g$ by
\[ 
g_j(x) = \begin{cases}
0 & \text{ if } b_{ij}\in \{0,\beta\}\ \forall i \in N\\
a_j f_j(x) & \text{ otherwise} 
\end{cases}
\]
Let $\x = \atp(\f, \B)$. Then $(\x,\g)$ is a price curve equilibrium.
\end{restatable}

\begin{restatable}{theorem}{thmPcToTp}
\label{thm:pc-to-tp}
Let $h$ be any constraint curve. Let $(\x, \g)$ be a price curve equilibrium, and define $\f$ and $\B$ by
\[ 
f_j(b) = \begin{cases}
h(b) & \text{ if } g_j \equiv 0\\
g_j(b) & \text{ otherwise} 
\end{cases}
\quad \quad \quad
b_{ij} = \begin{cases}
\beta & \text{ if } g_j \equiv 0 \text{ and } j \in R_i\\
0 & \text{ if } g_j \equiv 0 \text{ and } j \not\in R_i\\
x_{ij} & \text{ otherwise} 
\end{cases}
\] 
Then $\B$ is a Nash equilibrium of $\atp(\f)$.
\end{restatable}

\subsection{Equilibrium conditions for price curves and trading post}\label{sec:eq-cond}

Recall that $w_{ij} = 1$ if $j \in R_i$, and 0 otherwise. The following lemma for trading post states a useful necessary and sufficient condition for Nash equilibria of $\atp(\f)$.

\begin{restatable}{lemma}{lemTpEq}
\label{lem:tp-eq}
Let $\x = \atp(\f, \B)$. Then $\B \in NE(\atp(\f))$ if and only if all of the following hold:
\begin{enumerate}
\item For all $i \in N$, $x_{ij} = w_{ij} u_i(x_i)$ for all $j \in M$ where there exists $k \in N$ with $b_{kj} > 0$.
\item For all $i \in N$, $C_\f(b_i) = 1$.
\end{enumerate}
\end{restatable}

\begin{proof}
$(\implies)$ Assume that the two conditions of the lemma are true. First, we claim that $b_{ij} \in \{0,\beta\}$ for all $j \not \in R_i$: agent $i$ only spends money on goods in $R_i$. This is because $w_{ij} u_i(x_i) = 0$ for $j \not\in R_i$, but $b_{ij} > 0$ ensures that $x_{ij} > 0$, so $x_{ij} = w_{ij} u_i(x_i)$ would be impossible. Therefore $C_\f(b_i) = \sum_{j \in M} f_j(b_{ij}) = \sum_{j \in R_i} f_j(b_{ij}) = 1$.

Now suppose that $\B$ is not a Nash equilibrium: then there exists an agent $i$ and bid $b_i'$ such that $u_i(x_i') > u_i(x_i)$, where $\xprime$ is the resulting allocation when agent $i$ bids $b_i'$ and every agent $k \ne i$ still bids $b_k$. Condition 1 implies that $x_{ij} = u_i(x_i)$ for all $j \in R_i$ with $b_{ij} > 0$ (since $w_{ij} = 1$ for $j \in R_i$). Since $b_{ij} > 0$ only when $j \in R_i$, we have $x_{ij} = u_i(x_i)$ whenever $b_{ij} > 0$. Thus $u_i(x_i') > x_{ij}$ when $b_{ij} > 0$. Since $x_{ij}' \ge u_i(x_i')$ for all $j \in R_i$, we have $x'_{ij} > x_{ij}$ when $b_{ij} > 0$.

We next claim that $b'_{ij} > b_{ij}$ whenever $b_{ij} > 0$. If there exists $k \ne i$ with $b_{kj} > 0$, then $b'_{ij} > b_{ij}$ is necessary to ensure that $x'_{ij} > x_{ij}$. The only other possibility is that $b_{ij} > 0$, but $b'_{ij} = \beta$, and $b_{kj} \in \{0,\beta\}$ for all $k \ne i$. But in this case, following Step 1 of $\atp$'s allocation rule, $x_{ij} = s_j$. Then $u_i(x_i) = s_j$. This is the highest utility agent $i$ could ever have, since $u_i(x_i) \le s_j$ for all $j \in R_i$. This contradicts $u_i(x_i') > u_i(x_i)$. We conclude that $b'_{ij} > b_{ij}$ whenever $b_{ij} > 0$.

Therefore, since each $f_j$ is strictly increasing,
\[
\sum_{j \in M} f_j(b_{ij}) = \sum_{j: b_{ij} > 0} f_j(b_{ij}) < \sum_{j: b_{ij} > 0} f_j(b'_{ij}) \le \sum_{j \in M} f_j(b'_{ij})
\]
Since $C_\f(b_i) = \sum_{j \in M} f_j(b_{ij}) = 1$ by assumption, we have $\sum_{j \in M} f_j(b'_{ij}) > 1$. This means that $b_i'$ violates the bid constraint, and so is not a valid bid. Therefore $\B$ is a Nash equilibrium.

$(\impliedby)$ Suppose that $\B$ is a Nash equilibrium of $\atp(\f)$. If $C_\f(b_i) > 1$, $b_i$ violates the supply constraint, so $\B$ cannot be a Nash equilibrium. If $C_\f(b_i) < 1$, agent $i$ can improve her utility by bidding slightly more on every good (and thus receiving slightly more of every good). Thus $C_\f(b_i) = 1$ must hold. 

Suppose $x_{i\ell} \ne w_{i\ell} u_i(x_i)$ for some $\ell \in M$ where there exists $k \in N$ with $b_{k\ell} > 0$. By definition of $u_i$, $u_i(x_i) w_{i\ell} > x_{i\ell}$ is impossible, so we must have $x_{i\ell} > w_{i\ell} u_i(x_i)$. Consider a new bid $b_i'$ where $b'_{ij} = b_{ij}$ for all $j \ne \ell$, but $b'_{ij}$ is such that $x_\ell = w_{i\ell} u_i(x_i)$ (where $\xprime$ is the resulting allocation when $i$ bids $b_i'$ and each $k \ne i$ bids $b_k$). Thus $b'_{i\ell} < b_{i\ell}$.

By definition of $u_i$, we have $u_i(x_i') = u_i(x_i)$, but $C_\f(b'_i) < C_\f(b_i) = 1$, since $f_j(b'_{ij}) \le f_j(b_{ij})$ for all $j \in M$, and $f_\ell(b'_{i\ell}) < f_\ell(b_{i\ell})$. Thus there must exist a bundle $b_i''$ with $b_{ij}'' > b_{ij}'$ for all $j$, but $C_\f(b''_i) \le 1$, i.e., $b_i''$ obeys the bid constraint. Furthermore, let $\mathbf{x''}$ be the resulting allocation when $i$ bids $b''_i$ and each $k \ne i$ bids $b_k$: then $x''_{ij} > x'_{ij}$ for all $j \in M$. Therefore $u_i(x_i'') > u_i(x_i') = u_i(x_i)$. But this means $\B$ cannot be a Nash equilibrium, which is a contradiction.
\end{proof}

Next, we give an analogous lemma for price curve equilibrium. Note that the last condition in Lemma~\ref{lem:pc-eq} is simply one of the conditions in the definition of PCE.

\begin{restatable}{lemma}{lemPcEq}
\label{lem:pc-eq}
An allocation $\x$ and price curves $\g$ are a PCE if and only if all of the following hold:
\begin{enumerate}
\item For all $i \in N$, $x_{ij} = w_{ij} u_i(x_i)$ whenever $g_j \not\equiv 0$.
\item For all $i \in N$, $C_\g(x_i) = 1$.
\item For all $j \in M$, $\sum_{i \in N} x_{ij} \le s_j$, and $\sum_{i \in N} x_{ij} = s_j$ whenever $g_j\not \equiv 0$.
\end{enumerate}
\end{restatable}

\begin{proof}
The third condition is simply one of the two conditions in the definition of PCE. The other requirement for PCE is that $x_i \in D_i(\g)$ for all $i \in N$, so it suffices to show that $x_i \in D_i(\g)$ if and only if $C_\g(x_i) = 1$ and $x_{ij} = w_{ij} u_i(x_i)$ whenever $g_j\not\equiv 0$.

$(\implies)$ Suppose that $C_\g(x_i) = 1$, and $x_{ij} = w_{ij} u_i(x_i)$ whenever $g_j\not\equiv 0$. We first claim that agent $i$ only spends money on goods in $R_i$. This is because $w_{ij} u_i(x_i) = 0$ for $j \not \in R_i$ (because $w_{ij} = 0$ for $j \not \in R_i$), and spending money implies that $g_j\not\equiv 0$ and $x_{ij} > 0$, which makes $x_{ij} = w_{ij} u_i(x_i)$ impossible. Thus $C_\g(x_i) = \sum_{j \in R_i} g_j(x_{ij}) = \sum_{j \in R_i: g_j\not\equiv 0} g_j(x_{ij})$.

Now suppose for sake of contradiction that there exists another bundle $x'_i$ that is also affordable, and $u_i(x'_i) > u_i(x_i)$. For all $j \in R_i$ with $g_j \not\equiv 0$, we have $x_{ij} = w_{ij} u_i(x_i) = u_i(x_i)$ (because $w_{ij} = 1$ for $j \in R_i$), so $u_i(x'_i) > x_{ij}$ for $j \in R_i$, $g_j \not\equiv 0$. Therefore
\[
C_\g(x_i) = \sum_{j \in R_i: g_j\not\equiv 0} g_j(x_{ij}) < \sum_{j \in R_i: g_j\not\equiv 0} g_j(x'_{ij}) \le \sum_{j \in M} g_j(x'_{ij}) = C_\g(x'_i)
\]
Since, $C_\g(x_i) = 1$, we have $C_\g(x'_i) > 1$. But this implies that $x'_i$ is not affordable, which is a contradiction. Therefore $x_i \in D_i(\g)$.

$(\impliedby)$ Suppose $x_i \in D_i(\g)$. If $C_\g(x_i) > 1$, $x_i$ is not affordable, which is impossible. If $C_g(x_i) < 1$, agent $i$ can improve her utility by purchasing slightly more of every good. Thus $\sum_{j \in M} g_j(x_{ij}) = 1$ must hold. 

Suppose $x_{i\ell} \ne w_{i\ell} u_i(x_i)$ for some $\ell \in M$ where $g_j\not\equiv 0$. By definition, $u_i(x_i) w_{i\ell} > x_{i\ell}$ is impossible, so we must have $x_{i\ell} > w_{i\ell} u_i(x_i)$. Consider a bundle $x'_i$ where $x'_{ij} = x_{ij}$ for all $j \ne \ell$, but $x'_{i\ell} = w_{i\ell} u_i(x_i)$. Then $u_i(x'_i) = u_i(x_i)$. Furthermore, $g_\ell(x'_{i\ell}) < g_\ell(x_{i\ell})$, so $C_\g(x'_i) < C_\g(x_i) \le 1$. Consider another bundle $x''_i$ where $x''_{ij} > x'_{ij}$ for all $j \in M$, but $C_\g(x''_i) \le 1$: this is always possible because each $g_j$ is continuous, and $C_\g(x'_i) < 1$. Then $x''_i$ is affordable, but $u_i(x''_i) > u_i(x'_i) = u_i(x_i)$. This contradicts $x_i \in D_i(\g)$.
\end{proof}

We are now ready to move on to the reduction itself.

%% file: tp_to_pc.tex
\subsection{Transforming trading post equilibria into price curve equilibria}\label{sec:tp-to-pc}

This direction of the reduction will require an additional mild condition, involving the following definition.

\begin{definition}\label{def:homo}
We say that a function $f: \bbrpos \to \bbrpos$ is \emph{homogenous of degree $\alpha > 0$} if for any $b, c \in \bbrpos$, $f(c\cdot b) = c^{\alpha} f(b)$.
\end{definition}

Our main result of this section is the following theorem:

\thmTpToPc*

Before proving Theorem~\ref{thm:tp-to-pc}, we prove several helpful lemmas (Lemmas~\ref{lem:u-positive} -- \ref{lem:cost}). Throughout Lemmas~\ref{lem:u-positive} -- \ref{lem:cost}, we assume $\x, \g$, and $a_1\dots a_m$ are defined as in Theorem~\ref{thm:tp-to-pc}. We also assume that $\B \in NE(\atp(\f))$. Let $\xprime$ be the intermediate allocation after Step 2 of $\atp$'s allocation rule.

Our first lemma simply states that all agents end up with positive utility.
\begin{lemma}\label{lem:u-positive}
For all $i \in N$, $u_i(x_i) > 0$.
\end{lemma}

\begin{proof}
It is always possible for each agent to bid a nonzero amount on each good and obtain nonzero utility. Thus any Nash equilibrium must give each agent nonzero utility.
\end{proof}

The following lemma states that the intermediate allocation after Step 2 is in fact the final allocation.
\begin{lemma}\label{lem:x'}
We have $\x = \xprime$.
\end{lemma}

\begin{proof}
We need to show that Step 3 of $\atp$'s allocation rule is not invoked. Suppose it were invoked: then there is an agent $i$ who ends up with $x_{ij} = 0$ for all $j$, and thus $u_i(x_i) = 0$. But this contradicts Lemma~\ref{lem:u-positive}. We conclude that $\x = \xprime$.
\end{proof}

Lemma~\ref{lem:cost} states that under these constraint curves and bids, the bid constraint is equivalent to the price curves constraint.
\begin{lemma}\label{lem:cost}
For all $i \in N$, $C_\g(x_i) = C_\f(b_i)$.
\end{lemma}

\begin{proof}
By the allocation rule of $\atp$, for all $j \in M$ where there exists $k \in N$ with $b_{kj} > 0$, for all $i \in N$, we have
\[
x'_{ij} = \frac{b_{ij}}{\sum_{k \in N} b_{kj}} s_j
\]
Lemma~\ref{lem:x'} implies that $\x = \xprime$. Also, since $a_j = (\sum_{k \in N} b_{kj}/s_j)^{\alpha_j}$ we have $s_j /\sum_{k \in N} b_{kj} = a_j^{-1/\alpha_j}$, so
\[
x_{ij} = \frac{b_{ij}}{\sum_{k \in N} b_{kj}} s_j = b_{ij} a_j^{-1/\alpha_j}
\]
whenever there exists $k \in N$ with $b_{kj} > 0$. By the definition of $\g$, $g_j \not\equiv 0$ if and only if there exists $k \in N$ with $b_{kj} > 0$ (since constraint curves are assumed to be strictly increasing). Therefore
\begin{align*}
C_\g(x_i) =&\ \sum_{j \in M} g_j(x_{ij})\\
=&\ \sum_{j: g_j \equiv 0} g_j(x_{ij}) + \sum_{j: g_j \not \equiv 0} g_j(x_{ij})\\
=&\ \sum_{j: g_j \not \equiv 0} a_j f_j(x_{ij})\\
=&\ \sum_{j: g_j \not \equiv 0} a_j f_j(b_{ij} a_j^{-1/\alpha_j})\\
=&\ \sum_{j: g_j \not \equiv 0} \frac{a_j}{a_j} f_j(b_{ij})\\
=&\ \sum_{j: g_j \not \equiv 0} f_j(b_{ij})
\end{align*}
By the definition of $\g$, $g_j \equiv 0$ is equivalent to $b_{kj} \in \{0,\beta\}$ for all $k \in N$. Thus $\sum_{j: g_j \equiv 0} f_j(b_{ij}) = 0$, so
\begin{align*}
\sum_{j: g_j \not \equiv 0} f_j(b_{ij}) =&\ \sum_{j \in M} f_j(b_{ij})\\
=&\ C_\f(b_i)
\end{align*}
as required.
\end{proof}

We are now ready to prove the main result of this section.

\thmTpToPc*

\begin{proof}
Since $\B\in NE(\atp(\f))$, we have $C_\f(b_i) = 1$ for all $i \in N$ by Lemma~\ref{lem:tp-eq}. This implies $C_\g(x_i) = 1$ by Lemma~\ref{lem:cost}. Lemma~\ref{lem:tp-eq} also gives us $x_{ij} = w_{ij} u_i(x_i)$ whenever there exists $k \in N$ with $b_{kj} > 0$. As before, $g_j \not\equiv 0$ if and only if there exists $k \in N$ with $b_{kj} > 0$. Therefore $x_{ij} = w_{ij} u_i(x_i)$ whenever $g_j \not\equiv 0$.

Thus in order to apply Lemma~\ref{lem:pc-eq}, we just need to show that $\sum_{i \in N} x_{ij} \le s_j$ for all $j \in M$, and that $\sum_{i \in N} x_{ij} = s_j$ whenever $g_j\not\equiv 0$. Since $\x$ is a valid allocation, we immediately have $\sum_{i \in N} x_{ij} \le s_j$ for all $j \in M$. Consider an arbitrary good $j$ with $g_j \not\equiv 0$: then by the definition of $g_j$, there exists $k \in N$ with $b_{kj} > 0$. Thus good $j$ is allocated according to Step 1 of $\atp$'s allocation rule, and we get
\[
x'_{ij} = \frac{b_{ij}}{\sum_{k \in N} b_{kj}} \cdot s_j
\]
Summing this across agents gives us
\[
\sum_{i \in N} x'_{ij}= \sum_{i \in N} \frac{b_{ij}}{\sum_{k \in N} b_{kj}} \cdot s_j = s_j
\]
Thus by Lemma~\ref{lem:x'}, $\sum_{i \in N} x_{ij} = s_j$, as required. Therefore we can apply Lemma~\ref{lem:pc-eq} and conclude that $(\x, \g)$ is a PCE.
\end{proof}

%% file: pc_to_tp.tex
\subsection{Transforming price curve equilibria into trading post equilibria}\label{sec:pc-to-tp}

Our main result of this section is the following theorem:

\thmPcToTp*

The proof of this theorem is slightly more involved that the proof of Theorem~\ref{thm:tp-to-pc}, but the intuition is the same. As before, we prove this theorem via a series of lemmas (Lemmas~\ref{lem:u-positive-pc} -- \ref{lem:cost-pc}). Let $\xprime = \atp(\f, \B)$ be the final allocation resulting from bids $\B$, and let $\mathbf{x''}$ be the allocation resulting from bids $\B$ after Step 2 of $\atp$'s allocation rule. We use these definitions and assume that $(\x, \g)$ is a PCE for the remainder of Section~\ref{sec:pc-to-tp}.

As in the other direction of the reduction, our first lemma states that all agents end up with positive utility.
\begin{lemma}\label{lem:u-positive-pc}
For all $i \in N$, $u_i(x_i) > 0$.
\end{lemma}

\begin{proof}
Regardless of the price curves, it is always possible for each agent to buy a nonzero amount of each good and obtain nonzero utility. Since $x_i \in D_i(\g)$, $x_i$ must give agent $i$ nonzero utility.
\end{proof}

We next claim that for all goods with nonzero price, the intermediate allocation after Step 2 of $\atp(\f)$ is equal to $\x$, the allocation from the price curve equilibrium.
\begin{lemma}\label{lem:x''}
For all $i \in N$, $x''_{ij} = x_{ij}$ whenever $g_j \not\equiv 0$.
\end{lemma}

\begin{proof}
When $g_j \not\equiv 0$, $b_{ij} = x_{ij}$. Since each good is required by at least one agent, and $u_i(x_i) > 0$ for all $i$ by Lemma~\ref{lem:u-positive-pc}, there exists $k \in N$ where $x_{kj} > 0$. Therefore $b_{kj} > 0$, so we follow Step 1 of $\atp$'s allocation rule. Thus for all $i \in N$,
\begin{align*}
x''_{ij} =&\ \frac{b_{ij}}{\sum_{k \in N} b_{kj}} s_j\\
=&\ \frac{x_{ij}}{\sum_{k \in N} x_{kj}} s_j
\end{align*}
Since $(\x, \g)$ is a PCE, Lemma~\ref{lem:pc-eq} gives us $\sum_{k \in N} x_{kj} = s_j$ whenever $g_j\not\equiv 0$. Therefore
\[
x''_{ij} = \frac{x_{ij}}{\sum_{k \in N} x_{kj}} s_j = x_{ij}
\]
as required.
\end{proof}

The next lemma states that for all goods where some agent is bidding a positive amount, every agent's bundle in $\mathbf{x''}$ matches up exactly with her weights and her utility for $x_i$.
\begin{lemma}\label{lem:b-positive}
For all $j \in M$ where there exists $k \in N$ with $b_{kj} > 0$, we have $x''_{ij} = w_{ij} u_i(x_i)$.
\end{lemma}

\begin{proof}
By the definition of $\B$, if $b_{kj} > 0$ for some $k \in N$, then $g_j \not\equiv 0$. Since $(\x, \g)$ is a price curve equilibrium, we then have $x_{ij} = w_{ij} u_i(x_i)$ by Lemma~\ref{lem:pc-eq}. Lemma~\ref{lem:x''} gives us $x''_{ij} = x_{ij}$, so $x''_{ij} = w_{ij} u_i(x_i)$.
\end{proof}

Next, we show that each agent's utility for her bundle after Step 2 is equal to her utility for $x_i$.
\begin{lemma}\label{lem:u''}
For all $i \in N$, $u_i(x''_i) = u_i(x_i)$.
\end{lemma}

\begin{proof}
It suffices to show that for all $j \in R_i$, $x''_{ij} = u_i(x_i)$.

Case 1: $j \in R_i$ and $g_j \not\equiv 0$. Lemma~\ref{lem:x''} implies that $x''_{ij} = x_{ij}$ in this case. Since $(\x,\g)$ is a price curve equilibrium, Lemma~\ref{lem:pc-eq} implies that $x_{ij} = w_{ij} u_i(x_i)$. Since $w_{ij} = 1$ for $j \in R_i$, $x''_{ij} = u_i(x_i)$, as required.

Case 2: $j \in R_i$ and $g_j \equiv 0$. If $g_j \equiv 0$, the definition of $\B$ implies that all agents bid either $\beta$ or 0 on $j$. Thus will be following Step 2 of $\atp$'s allocation rule. By definition of $\B$, $b_{ij} = \beta$ in this case. Following Step 2 of the $\atp$ allocation rule, let $\ell_i$ be a good with $b_{i\ell_i} > 0$: then $x''_{ij} = x''_{i\ell_i}$. Since $b_{i\ell_i} > 0$ by assumption, we have $x''_{i\ell_i} = w_{i \ell_i} u_i(x_i)$ by Lemma~\ref{lem:b-positive}. Furthermore, $b_{i\ell_i} > 0$ implies $x''_{i\ell_i} > 0$, so $w_{i \ell_i} u_i(x_i) > 0$. Thus we must have $w_{i\ell_i} = 1$, which implies $x''_{ij} =x''_{i \ell_i}= u_i(x_i)$.

Therefore $x''_{ij} = u_i(x_i)$ for all $j \in R_i$, so $u_i(x''_i) = u_i(x_i)$, as required.
\end{proof}

The next lemma states that the intermediate allocation after Step 2 is equal to the final allocation produced by $\atp(\f, \B)$.
\begin{lemma}\label{lem:x'-pc}
We have $\xprime = \mathbf{x''}$.
\end{lemma}

\begin{proof}
We need to show that the penalty in Step 3 is not invoked. Suppose it is invoked: then there is a good $j$ allocated by Step 2 where $\sum_{i\in N} x''_{ij} > s_j$. For each $i \in N$ bidding $\beta$ on good $j$, define $\ell_i$ as usual: then $x''_{ij} = x''_{i\ell_i}$. Since $b_{i\ell_i} > 0$, Lemma~\ref{lem:b-positive} implies that $x''_{i\ell_i} = w_{i\ell_i} u_i(x_i)$ whenever $b_{ij} = \beta$.
\[
s_j < \sum_{i\in N} x''_{ij} = \sum_{i: b_{ij} = \beta} x''_{i\ell_i} = \sum_{i: b_{ij} = \beta} w_{i\ell_i} u_i(x_i)
\]
By definition of $\B$, we must have $g_j \equiv 0$: that is the only situation where agents bid $\beta$. Furthermore, $b_{ij} = \beta$ if and only if $j \in R_i$. Also using $w_{i\ell_i} \le 1$ (in reality, $w_{i\ell_i} = 1$ exactly, but we only need the inequality), gives us
\[
s_j < \sum_{i: j \in R_i} w_{i\ell_i} u_i(x_i) \le \sum_{i: j\in R_i} u_i(x_i)
\]
Using $w_{ij} = 1$ if and only if $j \in R_i$ then gives us
\[
s_j <  \sum_{i: j\in R_i} u_i(x_i) = \sum_{i: j \in R_i} w_{ij} u_i(x_i) = \sum_{i \in N} w_{ij} u_i(x_i)
\]
By definition of $u_i$, $x_{ij} \ge w_{ij} u_i(x_i)$ for all $j \in M$. Therefore
\[
\sum_{i \in N} x_{ij} \ge \sum_{i \in N} w_{ij} u_i(x_i) > s_j
\]
But this implies that $\x$ is not a valid allocation, which is a contradiction. We conclude that Step 3 is not invoked, and thus $u_i(x'_i) = u_i(x''_i)$, which is equal to $u_i(x_i)$ by Lemma~\ref{lem:u''}.
\end{proof}

Next, we show that the price curves constraint and bid constraint coincide.
\begin{lemma}\label{lem:cost-pc}
For all $i \in N$, $C_\f(b_i) = C_\g(x_i)$.
\end{lemma}

\begin{proof}
By the definition of $\B$, $b_{ij} \in \{0,\beta\}$ when $g_j \equiv 0$. Therefore:
\begin{align*}
C_\f(b_i) =&\ \sum_{j \in M} f_j(b_{ij})\\
 =&\ \sum_{j: g_j\not\equiv 0} f_j(b_{ij}) + \sum_{j: g_j \equiv 0} f_j(b_{ij})\\
=&\ \sum_{j: g_j\not\equiv 0} f_j(b_{ij})\\
=&\ \sum_{j: g_j\not\equiv 0} g_j(x_{ij})\\
=&\ \sum_{j \in M} g_j(x_{ij})\\
=&\ C_\g(x_i)
\end{align*}
\end{proof}

We are now ready to prove the main result of this section.

\thmPcToTp*

\begin{proof}
Suppose $(\x, \g)$ is a price curve equilibrium. By Lemma~\ref{lem:pc-eq}, we have $C_\g(x_i) = 1$ for all $i\in N$. Thus Lemma~\ref{lem:cost-pc} implies that $C_\f(b_i) = 1$ as well, which satisfies condition 2 of Lemma~\ref{lem:tp-eq}.

Lemma~\ref{lem:b-positive} implies that $x''_{ij} = w_{ij} u_i(x_i)$ whenever there exists $k \in N$ with $b_{kj} > 0$. Combining this with Lemmas~\ref{lem:u''} and \ref{lem:x'-pc} gives us $x'_{ij} = w_{ij} u_i(x'_i)$ whenever there exists $k \in N$ with $b_{kj} > 0$. This satisfies condition 1 of Lemma~\ref{lem:tp-eq}. Therefore by Lemma~\ref{lem:tp-eq}, $\B$ is a Nash equilibrium of $\atp(\f)$.

\end{proof}

%% file: scaling.tex
\section{Nash-implementing CES welfare functions with trading post}\label{sec:ces}

In this section, we use the reduction between price curves and augmented trading post to show that for any $\rho \in (-\infty, 1)$, $\atp(\rho)$ Nash-implements CES welfare maximization. Recall that $\atp(\rho)$ is the augmented trading post mechanism where $f_j(b) = b^{1-\rho}$ for all $j \in M$. Our key tools will be the reduction from Section~\ref{sec:reduction}, and pair of lemmas from~\cite{goel_beyond_2018} regarding price curve equilibria. The final result is Theorem~\ref{thm:ces}:

\begin{restatable}{theorem}{thmCES}
\label{thm:ces}
For any $\rho \in (-\infty, 1)$, the mechanism $\atp(\rho)$ Nash-implements the maximum CES welfare social choice rule.
\end{restatable}

Before we can prove Theorem~\ref{thm:ces}, we need one more property: Section~\ref{sec:scaling} shows that scaling the constraint curves does not affect the set of Nash equilibrium outcomes. We then prove the main theorem in Section~\ref{sec:ces-main}.

\subsection{Nash equilibria of trading post are invariant to scaling of constraint curves}\label{sec:scaling}

In order to use the reduction from Section~\ref{sec:reduction}, we would like to set $f_j(b) = g_j(b) = q_j b^{1-\rho}$. However, this would not be a valid mechanism: $q_1\dots q_m$ depend on the utility profile $\bfu$, and the mechanism cannot depend on $\bfu$. In this section, we show that scaling by $q_1\dots q_m$ does not affect the Nash equilibrium outcomes of $\atp$. This will allow us to use the mechanism $\atp(\rho)$ instead, which does not depend on $\bfu$.


Recall that for the mechanism $\atp(\f)$, $NE(\atp(\f))$ is the set of Nash equilibrium bids $\B$, and $NE_X(\atp(\f))$ is set of allocations $\x$ resulting from some $\B \in NE(\atp(\f))$.

\begin{lemma}\label{lem:scaling-intermediate}
Let $a_1\dots a_m$ be positive constants and let $\f$ be constraint curves where each $f_j$ is homogenous of degree $\alpha_j > 0$. Define $\fprime$ by $f'_j(b) = a_j f_j(b)$. Then $NE_X(\atp(\f)) \subseteq NE_X(\atp(\fprime))$.
\end{lemma}

\begin{proof}
Let $\x$ be an arbitrary allocation in $NE_X(\atp(\f))$; we will show that $\x \in NE_X(\atp(\fprime))$. By definition, there exist bids $\B \in NE(\atp(\f))$ such that $\x = \atp(\f, \B)$. Define $\bprime$ by $b'_{ij} = a_j^{-1/\alpha_j} b_{ij}$ when $b_{ij} > 0$ and $b'_{ij} = b_{ij}$ otherwise. We first show that $C_\fprime(\bprime) = C_\f(\B)$:
\[
\sum_{j \in M} f'_j(b'_{ij}) = \sum_{j \in M} a_j f_j(a_j^{-1/\alpha_j} b_{ij})= \sum_{j \in M} a_j (a_j^{-1/\alpha_j})^{\alpha_j} f_j(b_{ij}) =\sum_{j \in M} f_j(b_{ij})
\]
Let $\xprime = \atp(\fprime, \bprime)$. For any good $j$ where $b'_{kj} > 0$ for some $k \in N$ (and thus also $b_{kj} > 0$),
\begin{align*}
x'_{ij} =&\ \frac{b'_{ij}}{\sum_{k \in N} b'_{kj}}\\
=&\ \frac{a_j^{-1/\alpha_j} b_{ij}}{\sum_{k \in N} a_j^{-1/\alpha_j} b_{kj}}\\
=&\ \frac{b_{ij}}{\sum_{k \in N} b_{kj}}\\
=&\ x_{ij}
\end{align*}
Thus for any good $j$ where $b'_{kj} > 0$ for some $k \in N$, we have $x'_{ij} = x_{ij}$. For any good $j$ where $b'_{kj} \in \{0,\beta\}$ for all $k$, we also have $b_{kj} \in \{0,\beta\}$ for all $k$. Thus in both cases we follow Step 2 of $\atp$'s allocation rule. Since $x'_{ij} = x_{ij}$ for the good $j$ where $b'_{kj} > 0$ for some $k$, Step 2 results in $x'_{ij} = x_{ij}$ for goods where $b'_{kj} \in \{0,\beta\}$ for all $k$. Therefore $\xprime = \x$.

This implies that $u_i(x_i) = u_i(x'_i)$ for all $i \in N$. Therefore $x_{ij} = w_{ij} u_i(x_i)$ if and only if $x'_{ij} = w_{ij} u_i(x'_i)$. Thus the conditions of Lemma~\ref{lem:tp-eq} hold for $\B, \f$ if and only if they hold for $\bprime, \fprime$. Therefore since $\B \in NE(\atp(\f))$, we have $\bprime \in NE(\atp(\fprime))$, and thus $\x = \xprime \in NE_X(\atp(\fprime))$. We conclude that $NE_X(\atp(\f)) \subseteq NE_X(\atp(\fprime))$.
\end{proof}

\begin{lemma}\label{lem:scaling}
Let $a_1\dots a_m$ be positive scalars and let $\f$ be constraint curves where each $f_j$ is homogenous of degree $\alpha_j$. Define $\fprime$ by $f'_j(b) = a_j f_j(b)$. Then $NE_X(\atp(\f)) = NE_X(\atp(\fprime))$.
\end{lemma}

\begin{proof}
Lemma~\ref{lem:scaling-intermediate} gives us $NE_X(\atp(\f)) \subseteq NE_X(\atp(\fprime))$, so it remains only to show that $NE_X(\atp(\fprime)) \subseteq NE_X(\atp(\f))$. We can actually do this by symmetry. Define $a'_1\dots a'_m$ by $a_j' = 1/a_j$. Then $a'_1\dots a'_m$ are positive scalars such that $f_j(b) = a'_j f'_j(b)$. Each $f_j'$ is also be homogenous of degree $\alpha_j$:
\[
f'_j(c\cdot b) = a_j f_j(c \cdot b) = c^{\alpha_j} a_j f_j(b) = c^{\alpha_j} f_j'(b)
\]
Then we can apply Lemma~\ref{lem:scaling-intermediate} with the roles of $\fprime$ and $\f$ swapped to give us $NE_X(\atp(\fprime)) \subseteq NE_X(\atp(\f))$, which completes the proof.
\end{proof}

%% file: ces.tex
\subsection{Main theorem}\label{sec:ces-main}

The last tool we need is the following pair of lemmas, both due to~\cite{goel_beyond_2018}:

\begin{lemma}[\cite{goel_beyond_2018}]\label{lem:ces-price-curves}
For utility profile $\bfu$, $\rho \in (-\infty, 1)$, and $\x \in \Psi_\rho(\bfu)$, there exist price curves $\g$ such that $(\x,\g)$ is a price curve equilibrium. Furthermore, for each $j \in M$, $g_j$ takes the form $g_j(x) = q_j x^{1-\rho}$ for some nonnegative constants $q_1\dots q_m$.
\end{lemma}

\begin{lemma}[\cite{goel_beyond_2018}]\label{lem:ces-sufficient}
Suppose $\rho \in (-\infty, 1)$ and that price curves $\g$ take the form $g_j(x) = q_j x^{1-\rho}$ for each $j \in M$, for some nonnegative constants $q_1\dots q_m$. Then if $(\x, \g)$ is a PCE, $\x \in \Psi_\rho(\bfu)$.
\end{lemma}

Lemma~\ref{lem:ces-price-curves} states that for any maximum CES welfare allocation $\x$, there exist price curves of the form $g_j(x) = q_j x^{1-\rho}$ where $(\x, \g)$ is a PCE. Lemma~\ref{lem:ces-sufficient} states the converse: if $\g$ takes the form $g_j(x) = q_j x^{1-\rho}$ for nonnegative constants $q_1\dots q_m$, and $(\x, \g)$ is a PCE, then $\x$ is a maximum CES welfare allocation. Furthermore, $q_1\dots q_m$ are the Lagrange multipliers of the convex program for maximizing CES welfare, so $q_1\dots q_m$ can be computed in polynomial time.

We are now finally ready to prove that $\atp(\rho)$ Nash-implements $\Psi_\rho$. We will make of the follow results from previous sections:
\begin{enumerate}
\item Theorem~\ref{thm:tp-to-pc}: Any Nash equilibrium of $\atp$ can be converted into an ``equivalent" price curve equilibrium.
\item Theorem~\ref{thm:pc-to-tp}: Any price curve equilibrium can be converted into an ``equivalent" Nash equilibrium of $\atp$.
\item Lemma~\ref{lem:scaling}: The set of Nash equilibrium outcomes of $\atp$ is invariant to constant scaling of the constraint curves.
\item Lemma~\ref{lem:ces-price-curves}~\cite{goel_beyond_2018}: For any maximum CES welfare allocation $\x,$ there exist price curves $\g$ of the form $g_j(x) = q_j x^{1-\rho}$ such that $(\x, \g)$ is a PCE.
\item Lemma~\ref{lem:ces-sufficient}~\cite{goel_beyond_2018}: If there exist price curves $\g$ of the form $g_j(x) = q_j x^{1-\rho}$ such that $(\x, \g)$ is a PCE, then $\x$ is a maximum CES welfare allocation.
\end{enumerate}

Recall that for a utility profile $\bfu$, the induced game of mechanism $\atp(\rho)$ is denoted by $\atp(\rho)(\bfu)$. We left $\bfu$ implicit when dealing with Nash equilibria in previous sections, but we make it explicit here.

\thmCES*

\begin{proof}
We need to show that for any utility profile $\bfu$, $\emptyset \ne NE_X(\atp(\rho)(\bfu)) \subseteq \Psi_\rho(\bfu)$: in words, for any $\bfu$, there is at least one Nash equilibrium, and every Nash equilibrium allocation of $\atp(\rho)(\bfu)$ is a maximum CES welfare allocation with respect to $\rho$ and $\bfu$.

Pick any $\xs \in \Psi_\rho(\bfu)$, and define $q_1\dots q_m$ and $\g$ as in Lemma~\ref{lem:ces-price-curves}. Define $\f$ by $f_j(b) = q_j b^{1-\rho}$ when $q_j \ne 0$, and $f_j(b) = b^{1-\rho}$ when $q_j = 0$. By Lemma~\ref{lem:scaling}, we have $NE_X(\atp(\rho)(\bfu)) = NE_X(\atp(\f)(\bfu))$. Thus it suffices to show that $\emptyset \ne NE_X(\atp(\f)(\bfu)) \subseteq \Psi_\rho(\bfu)$.

We first show that $NE_X(\atp(\f)(\bfu)) \ne \emptyset$, i.e., $\atp(\f)(\bfu)$ has at least one Nash equilibrium. By Lemma~\ref{lem:ces-price-curves}, $(\xs, \g)$ is a PCE. Since $g_j(x) = q_j x^{1-\rho}$ by Lemma~\ref{lem:ces-price-curves}, we have $f_j(b) = g_j(b)$ whenever $q_j \ne 0$ (which is equivalent to $g_j \equiv 0$). When $g_j \equiv 0$, $f_j(b) = b^{1-\rho}$, which is strictly increasing. Thus $\f$ satisfies the requirements of Theorem~\ref{thm:pc-to-tp}. If we define $\B$ as a function of $\xs$ as in Theorem~\ref{thm:pc-to-tp}, then by Theorem~\ref{thm:pc-to-tp}, $\B \in NE(\atp(\f)(\bfu))$. Therefore $\atp(\f)(\bfu)$ has at least one Nash equilibrium.

It remains to show that $NE_X(\atp(\f)(\bfu)) \subseteq \Psi_\rho(\bfu)$, i.e., every Nash equilibrium outcome of $\atp(\f)(\bfu)$ is a maximum CES welfare allocation. Consider an arbitrary $\x \in NE_X(\atp(\f)(\bfu))$. Then there exists $\B \in NE(\atp(\f)(\bfu))$ such that $\x = \atp(\f, \B)$. Noting that each $f_j$ is homogenous of degree $1-\rho$, define $\gprime$ as a function of $\f$ and $a_1\dots a_m$ as a function of $\B$ as in Theorem~\ref{thm:tp-to-pc}:
\[
a_j = \Big(\frac{\sum_{k \in N} b_{kj}}{s_j}\Big)^{1-\rho}
\quad \text{and} \quad
g'_j(x) = \begin{cases}
0 & \text{ if } b_{ij}\in \{0,\beta\}\ \forall i \in N\\
a_j f_j(x) & \text{ otherwise}
\end{cases} 
\]
By Theorem~\ref{thm:tp-to-pc}, $(\x, \gprime)$ is a PCE. Furthermore, we can write each $g'_j$ as $g'_j(x) = q'_j x^{1-\rho}$ for nonnegative constants $q'_1\dots q'_m$. Therefore by Lemma~\ref{lem:ces-sufficient}, $\x \in \Psi_\rho(\bfu)$.

Thus we have shown that $\x \in \Psi_\rho(\bfu)$ for all $\x \in NE_X(\atp(\f)(\bfu))$, so $NE_X(\atp(\rho)(\bfu)) \subseteq \Psi_\rho(\bfu)$. Since $NE_X(\atp(\rho)(\bfu)) = NE_X(\atp(\f)(\bfu))$, we conclude that $\emptyset \ne NE_X(\atp(\rho)(\bfu)) \subseteq \Psi_\rho(\bfu)$.
\end{proof}

Finally, we note that a Nash equilibrium $\B \in NE(\atp(\rho)(\bfu))$ can be computed in polynomial time. Since $q_1\dots q_m$ are the Lagrange multipliers of the convex program for maximizing CES welfare, they can be computed in polynomial time. Then Theorem~\ref{thm:pc-to-tp} can be applied to obtain $\mathbf{b'} \in NE(\atp(\f)(\bfu))$, and finally Lemma~\ref{lem:scaling} yields an equivalent $\B \in NE(\atp(\rho)(\bfu))$.

%% file: veto.tex
\subsubsection{Maskin's approach and no veto power}\label{sec:no-veto-power}

As discussed in Section~\ref{sec:related}, Maskin proved that in a very general environment, any social choice rule satisfying monotonicity and no veto power is Nash-implementable~\cite{maskin_nash_1999}. We briefly show that bandwidth allocation does not satisfy no veto power for any $\rho \in (-\infty, 1]$, and thus is not conducive to Maskin's approach.

\begin{definition}\label{def:veto}
A social choice rule $\Psi$ satisfies \emph{no veto power} if whenever there exists an allocation $\x$ where for all $i \in N$ except at most 1, $u_i(x_i) \ge u_i(y_i)$ for all allocations $\y$, we have $\x \in \Psi(\bfu)$.
\end{definition}

In words, if there is a single allocation that everyone (except at most one agent) agrees is their favorite, then that allocation should be optimal under $\Psi$ (the last agent should not be able to ``veto" this allocation). In general, agents will not agree on a favorite allocation: each agent would like to receive all of the resources herself. However, when agents' $R_i$ sets are pairwise disjoint, it is possible for all agents to agree on a favorite allocation.

Consider an instance with $n$ agents and $n$ goods, each with supply 1. For all $i \in N$, let $R_i = \{i\}$: each agent just desires a single good. Consider the allocation $\x$ where for all $i \in \{1\dots n-1\}$, $x_{ii} = 1$, but $x_{nn} = 0$ (and $x_{ij} = 0$ otherwise). For agents $1\dots n-1$, this is the most utility they can possibly get, so this satisfies the precondition of Definition~\ref{def:veto}. However, for any $\rho \in (-\infty, 1]$, $\x \not \in \Psi_\rho(\bfu)$, because the CES welfare can be improved by increasing $x_{nn}$. Specifically, for every $\rho$, the unique optimal CES allocation has $x_{ii} = 1$ for all $i \in \{1\dots n\}$.

%% file: dse_intro.tex
\section{Dominant strategy implementation, strategyproofness, and maxmin welfare}\label{sec:dse}

In Section~\ref{sec:ces}, we showed that for every $\rho \in (-\infty, 1)$, CES welfare maximization is Nash-implementable. A natural question to ask is whether this result can be improved to dominant strategy equilibrium implementation (DSE implementation, for short). In Section~\ref{sec:ces-not-sp}, we show that for every $\rho \in (-\infty, 1]$, $\Psi_\rho$ is not DSE-implementable (Theorem~\ref{thm:ces-not-sp}). In contrast, Section~\ref{sec:maxmin} shows that maxmin welfare ($\rho = -\infty$) is in fact DSE-implementable by a simple revelation mechanism (Theorem~\ref{thm:maxmin}). Finally, Section~\ref{sec:maxmin-nash} uses a more complex mechanism to show that maxmin welfare is also Nash-implementable (Theorem~\ref{thm:maxmin-nash}).

\paragraph{Review of relevant concepts.}

An important property related to DSE implementation is strategyproofness. Recall that a mechanism is strategyproof when honestly reporting one's preferences is always a dominant strategy. As discussed in Section~\ref{sec:impl}, DSE-implementability implies strategyproofness by the revelation principle, but the converse is not necessarily true: strategyproofness ensures that truth-telling is \emph{a} dominant strategy equilibrium, but there could also be bad dominant strategy equilibria. For our positive result DSE result (Theorem~\ref{thm:maxmin}), we will show that our mechanism is strategyproof, and also that there are no bad dominant strategy equilibria. For our negative DSE result (Theorem~\ref{thm:ces-not-sp}), we show that the social choice rule in question is not strategyproof, which implies that it is not DSE-implementable.

Furthermore, as also discussed in Section~\ref{sec:impl}, DSE-implementability does \emph{not} imply Nash-implementability: DSE-implementability requires every DSE to be consistent with $\Psi$, but the mechanism might have additional (non-DSE) Nash equilibria that are not consistent with $\Psi$. In fact, our DSE implementation of maxmin welfare is not a Nash implementation: it may have Nash equilibria that are not optimal (see Section~\ref{sec:maxmin-not-nash}). In Section~\ref{sec:maxmin-nash}, we give a more complex mechanism that does Nash-implement maxmin welfare (Theorem~\ref{thm:maxmin-nash}).

We briefly discuss a subtlety relating to uniqueness (and lack thereof). In a sense, all strategyproof mechanisms that implement a social choice rule $\Psi$ are the same: they all ask agents to report their utility functions $u_1\dots u_n$, then compute an outcome $\x \in \Psi(\bfu)$\footnote{This is in the setting where no payments are involved, like this paper. If payments are allowed, these mechanisms can of course differ in what agents are asked to pay.}. However, if $\Psi(\bfu)$ contains multiple elements (i.e., there are multiple optimal allocations), it may matter which is chosen. If leftover supply is allocated arbitrarily, it can be hard to reason about the optimal allocation under different utility profiles. Furthermore, not even the optimal vector of agent utilities is unique for maxmin welfare and utilitarian welfare (although it is for $\rho \in (-\infty, 1))$.

Consequently, for both of our positive results (Theorems~\ref{thm:maxmin} and \ref{thm:maxmin-nash}), we will specify our mechanism such that it selects a unique allocation for each utility profile $\bfu$ (even when they are multiple optimal allocations). For our negative result (Theorem~\ref{thm:ces-not-sp}), we will give an instance where an agent lying makes her utility in \emph{every} new optimal allocation strictly larger than her utility in \emph{every} optimal allocation under a truthful report.

\subsection{For all $\rho \in (-\infty, 1]$, CES welfare maximization is not DSE-implementable}\label{sec:ces-not-sp}

To show impossibility of DSE implementation, it is sufficient to show impossibility of strategyproofness. Our counterexample will be the following instance with 5 agents and 7 goods, where each row is an agent, each column is a good, and the cell in the $i$th row and $j$th column gives $w_{ij}$:
\begin{table}[h]
\centering
\begin{tabular}{ c|ccccccc} 
& g1 & g2 & g3 & g4 & g5 & g6 & g7\\
\hline
agent 1 & 1 & 1 & 0 & 0 & 0 & 0 & 1\\
agent 2 & 0 & 0 & 1 & 1 & 0 & 0 & 1\\
agent 3 & 0 & 0 & 0 & 0 & 1 & 1 & 1\\
agent 4 & 1 & 0 & 1 & 0 & 1 & 0 & 0\\
agent 5 & 0 & 1 & 0 & 1 & 0 & 1 & 0\\
\end{tabular}
\end{table}

Let the supply of good 7 be 2, and let all other goods have supply 1. Notice that agents 1, 2, and 3 all conflict on good 7, but otherwise are not in competition. Agents 4 and 5 are not in competition with each other, but each conflicts with each of agents 1, 2, and 3. Let $\bfu$ denote this utility profile, and $\mathbf{u'}$ denote the utility profile where $R'_4 = \{1,3,5,7\}$ instead of $R_4 = \{1,3,5\}$, and all other utilities are unchanged. We will claim that under utility profile $\bfu$, agent 4 can increase her utility by misreporting $R'_4$ instead of $R_4$.

We will prove this using two main lemmas. Lemma~\ref{lem:truthful} states that when agent 4 truthfully reports $R_4$, her utility is strictly less than 1/2. Lemma~\ref{lem:lie} states that when agent 4 lies and reports $R'_4$ instead, her utility is at least 1/2 (Lemma~\ref{lem:concave-max} is a tool used in the proof of Lemma~\ref{lem:lie}). Note that each lemma is referring to agent 4's \emph{true} utility function $u_4$.

\begin{lemma}\label{lem:truthful}
For every $\rho \in (-\infty, 1]$, every $\x \in \Psi_\rho(\bfu)$ has $u_4(x_4) < 1/2$.
\end{lemma}

\begin{proof}
For $\rho = 0$, an optimal Nash welfare allocation can be computed explicitly, and any such allocation $\x$ will have $u_4(x_4) < 1/2$. Recall that although the optimal allocation may not be unique, the optimal utility vector is, since Nash welfare is strictly concave.

Let $A = \{1,2,3\}$ and $B = \{4,5\}$.  For $\rho \in (-\infty, 0)\cup(0, 1]$, we write the following convex program for maximizing CES welfare:
\begin{align*}
\max\limits_{u_1, u_2, u_3, u_4, u_5 \in \bbrpos} &\ (u_1^\rho + u_2^\rho + u_3^\rho + u_4^\rho + u_5^\rho)^{1/\rho} \\\nonumber
s.t.\ &\ u_i + u_k \le 1\ \ \forall i \in A, k \in B\\
&\ u_1 + u_2 + u_3 \le 2
\end{align*}
We are using $u_i$ as a variable in the convex program, but we will reserve $u_i(x_i)$ to denote agent $i$'s utility for $\x \in \Psi_\rho(\bfu)$. 

By construction, for every $i \in A$ and $k \in B$, there is a good such $j$ such that $x_{ij} + x_{kj} \le 1$ where $w_{ij} = w_{kj} = 1$. This implies that for any $\x \in \Psi_\rho(\bfu)$, $u_i(x_i) + u_k(x_k) \le 1$ for all $i \in A$ and $k \in B$. Similarly, the supply constraint of good 7 implies that $u_1(x_1) + u_2(x_2) + u_3(x_3) \le 2$. Furthermore, any such allocation is indeed feasible: simply let $x_{ij} = w_{ij} u_i(x_i)$ for all $i \in N$. This means that the set of possible utilities for feasible allocations is equal to the set of feasible solutions $u_1\dots u_5$ to the above convex program. Thus this program correctly maximizes CES welfare. This implies that for every $\x \in \Psi_\rho(\bfu)$, there exists an optimal solution to the above convex program $u_1^* \dots u_5^*$ such that $u_i(x_i) = u_i^*$ for all $i \in N$. We proceed by case analysis.

Case 1: $u_1^* + u_2^* + u_3^* = 2$. In this case, one of those three agents must have utility at least 2/3. Since agent 4 is in competition with each of those agents for a good with supply 1, this implies that $u_4^* = 1/3 < 1/2$.

Case 2: $u_1^* + u_2^* + u_3^* \ne 2$. Because of the convex program's second constraint, $u_1^* + u_2^* + u_3^* > 2$ is not feasible, so we must have $u_1^* + u_2^* + u_3^* < 2$. In this case, the convex program above reduces to:
\begin{align*}
\max\limits_{u_1, u_2, u_3, u_4, u_5 \in \bbrpos} &\ (u_1^\rho + u_2^\rho + u_3^\rho + u_4^\rho + u_5^\rho)^{1/\rho} \\\nonumber
s.t.\ &\ u_i + u_k \le 1\ \ \forall i \in A, k \in B
\end{align*}
We next claim that there exists $u_A^*$ and $u_B^*$ such that $u_i^* = u_A^*$ for all $i \in A$, and $u_k^* = u_B^*$ for all $i \in B$. Suppose there exists $i, i' \in A$ such that $u_i^* < u_{i'}^*$. Then we could increase $u_1^*$ while still obeying the constraints of this program, and increase $\frac{1}{\rho}(u_1^\rho + u_2^\rho + u_3^\rho + u_4^\rho + u_5^\rho)$. That would imply that $u_1^* \dots u_5^*$ could not be optimal. Furthermore, an identical argument applies to the agents in $\B$.

Therefore there exists $u_A^*$ and $u_B^*$ such that $u_i^* = u_A^*$ for all $i \in A$, and $u_k^* = u_B^*$ for all $i \in B$, so we can further rewrite the convex program as
\begin{align*}
\max\limits_{u_A, u_B\in \bbrpos} &\ (3 u_A^\rho + 2 u_B^\rho)^{1/\rho} \\\nonumber
s.t.\ &\ u_A + u_B \le 1
\end{align*}
where $(u_A^*, u_B^*)$ is the optimal solution of this program. Clearly we must have $u_B^* = 1- u_A^*$. Furthermore, we claim that we can change the objective function from $(3 u_A^\rho + 2 u_B^\rho)^{1/\rho}$ to $\frac{1}{\rho}(3 u_A^\rho + 2 u_B^\rho)$, and that this changes the optimal value of the program, but does not change the optimal solution, i.e., the argmax. This because when $\rho > 0$, maximizing $(3 u_A^\rho + 2 u_B^\rho)^{1/\rho}$ is equivalent to maximizing $3 u_A^\rho + 2 u_B^\rho$, and when $\rho < 0$, maximizing $(3 u_A^\rho + 2 u_B^\rho)^{1/\rho}$ is equivalent to minimizing $3 u_A^\rho + 2 u_B^\rho$, which is equivalent to maximizing $\frac{1}{\rho}(3 u_A^\rho + 2 u_B^\rho)$. Thus our new convex program is
\begin{align*}
\max\limits_{u_A \in [0,1]} &\ \frac{1}{\rho}(3 u_A^\rho + 2 (1-u_A)^\rho)\nonumber
\end{align*}
This is a program we can analyze. For $\rho = 1$, the objective function becomes $u_A + 2$, so we immediately have $u_A^* = 1$ and thus $u_4^* = u_B^*= 0$. For $\rho < 1$, we take the derivative with respect to $u_A$ should be 0 when evaluated at $u_A^*$:
\begin{align*}
3{u_A^*}^{\rho - 1} - 2(1- u_A^*)^{\rho - 1} =&\ 0\\
3{u_A^*}^{\rho - 1} =&\ 2(1- u_A^*)^{\rho - 1}\\
3^{\frac{1}{\rho - 1}}u_A^* =&\ 2^{\frac{1}{\rho - 1}}(1- u_A^*)\\
(3^{\frac{1}{\rho - 1}} + 2^{\frac{1}{\rho - 1}})u_A^* =&\ 2^{\frac{1}{\rho - 1}}\\
u_A^* =&\ \frac{2^{\frac{1}{\rho - 1}}}{3^{\frac{1}{\rho - 1}} + 2^{\frac{1}{\rho - 1}}}\\
u_A^* =&\ \frac{1}{(3/2)^{\frac{1}{\rho-1}} + 1}
\end{align*}
Since $\rho < 1$, $\rho - 1$ is negative, so $\frac{1}{\rho - 1}$ is negative. Then since $3/2 > 1$, $(3/2)^{\frac{1}{\rho - 1}} < 1$. Altogether, this implies that
\[
u_1^* = u_2^* = u_3^* = u_A^* > 1/2 \quad \text{and} \quad u_4^* = u_5^* = u_B^* < 1/2
\]
as required.
\end{proof}

The following lemma is a standard property of strictly concave and differentiable functions: it essentially states that any such function is bounded above by any tangent line. This lemma is sometimes called the ``Rooftop Theorem". To avoid confusion with $\mathbf{u'}$ and $\xprime$, we use $D h$ to denote the derivative of $h$, instead of $h'$ (this is Euler's notation for derivatives).

\begin{lemma}\label{lem:concave-max}
Let $h: \bbr \to \bbr$ be strictly concave and differentiable. Then for all $a,b \in \bbr$ where $a\ne b$, $h(a) < h(b) + (Dh(b))(a-b)$.
\end{lemma}

\begin{lemma}\label{lem:lie}
For any $\rho \in (-\infty, 1]$, for any $\xprime \in \Psi_\rho(\mathbf{u'})$, $u_4(x'_4) \ge 1/2$.
\end{lemma}

\begin{proof}
It suffices to show that $u'_4(x'_4) \ge 1/2$: since $R_4 \subset R'_4$, we have
\[
u_4(x'_4) = \min_{j \in R_4} x'_{4j} \ge \min_{j \in R'_4} x'_{4j} = u'_4(x'_4)
\]
For $\rho = 1$, the set of optimal allocations can be computed explicitly to see that for all $\xprime \in \Psi_{\rho = 1}(\mathbf{u'})$, $u_i'(x'_i) = 1/2$ for all $i \in N$. 

We now use this to show that for any $\rho \in (-\infty, 1)$, for any $\xprime \in \Psi_\rho(\mathbf{u'})$, $u'_i(x'_i) = 1/2$ for all $i \in N$. Intuitively, the larger $\rho$ is, the more we care about efficiency and the less we care about fairness. But if the most efficient solution (i.e., the optimal allocation for $\rho = 1$) coincides with the most fair solution (i.e., having all utilities equal), then no matter how much we care about efficiency vs fairness, we should get the same outcome.

Let $\xs$ be any allocation in $\Psi_{\rho = 1}(\mathbf{u'})$: then $u_i'(x_i^*) = 1/2$ for all $i \in N$. Fix a $\rho \in (-\infty, 1)$ and let $h(a) = \frac{1}{\rho} a^\rho$ if $\rho \ne 0$, and $h(a) = \log(a)$ if $\rho = 1$. In both of these $h$ is strictly concave and differentiable. Consider any allocation $\y$ where for some $i \in N$, $u_i(y_i) \ne 1/2$. For brevity, let $u_i'^* = u_i'(x_i^*)$ and $u_i'^y = u_i'(y_i)$.

For all such $i$, Lemma~\ref{lem:concave-max} implies that $h(u_i'^y) < h(u_i'^*) + (Dh(u_i'^*))(u_i'^y - u_i'^*)$ for $u_i'^* \ne u_i'^y$. When $u_i'^y = 1/2 = u_i'^*$, we have $h(u_i'^y) = h(u_i'^*) + (Dh(u_i'^*))(u_i'^y - u_i'^*) = 0$. Thus for $\rho \ne 0$, we have
\begin{align*}
 \sum_{i \in N}  \frac{1}{\rho}{(u_i'^y)}^\rho =&\ \sum_{i \in N} h(u_i'^y)\\
<&\ \sum_{i \in N} \Big(h(u_i'^*) + (Dh(u_i'^*))(u_i'^y - u_i'^*)\Big)\\
=&\ \sum_{i \in N} \Big( \frac{1}{\rho} {(u_i'^*)}^\rho + (Dh(1/2))(u_i'^y - u_i'^*)\Big)\\
=&\ \frac{1}{\rho} \sum_{i \in N} {(u_i'^*)}^\rho +\sum_{i \in N} (Dh(1/2))(u_i'^y - u_i'^*)\\
=&\ \frac{1}{\rho} \sum_{i \in N} {(u_i'^*)}^\rho +(Dh(1/2))\Big(\sum_{i \in N}u_i'^y - \sum_{i \in N} u_i'^*\Big)
\end{align*}
Since $\xs \in \Psi_{\rho = 1}(\mathbf{u'})$, $\sum_{i \in N} u_i'^* \ge \sum_{i \in N}u_i'^y$. Therefore $(Dh(1/2))\big(\sum_{i \in N}u_i'^y - \sum_{i \in N} u_i'^*\big) \le 0$, so
\begin{align*}
\frac{1}{\rho} \sum_{i \in N}  {(u_i'^y)}^\rho < \frac{1}{\rho} \sum_{i \in N} {(u_i'^*)}^\rho
\end{align*}
As before, this implies that $\big(\sum_{i \in N}  {(u_i'^y)}^\rho\big)^{1/\rho} < \big(\sum_{i \in N} {(u_i'^*)}^\rho\big)^{1/\rho}$. The analysis for $\rho = 0$ (i.e., Nash welfare) is the same, except we end up with $\sum_{i \in N}  \log(u_i'^y) < \sum_{i \in N} \log(u_i'^*)$ instead, which implies $\prod_{i \in N} u_i'^y < \prod_{i \in N} u_i'^*$.

Thus for any allocation $\y$ where there exists $i \in N$ with $u_i'(y_i) \ne 1/2$, the CES welfare of $\xs$ is better than the CES welfare of $\y$. This implies that for any $\rho \in (-\infty, 1]$, any $\xprime \in \Psi_\rho(\mathbf{u'})$ must have $u_i'(x'_i) = 1/2$ for all $i \in N$.
\end{proof}

\begin{theorem}\label{thm:ces-not-sp}
For all $\rho \in (-\infty, 1]$, $\Psi_\rho$ is not DSE-implementable.
\end{theorem}

\begin{proof}
Similar to the proof of Theorem~\ref{thm:maxmin}, it suffices to show that $\Psi_\rho$ cannot be computed in a strategyproof mechanism. Suppose there were a strategyproof mechanism $H$: then for utility profile $\bfu$, $H$ must return an allocation $\x \in \Psi_\rho(\bfu)$, and for utility profile $\mathbf{u'}$, $H$ must return an allocation $\xprime \in \Psi_\rho(\mathbf{u'})$. By Lemmas~\ref{lem:truthful} and \ref{lem:lie}, we have $u_4(x_4) < 1/2$ and $u_4(x'_4) \ge 1/2$. If agent 4 reports $R'_4 = \{1,3,5,7\}$ instead of $R_4 = \{1,3,5\}$, she alters the utility profile from $\bfu$ to $\mathbf{u'}$, which resulting in her receiving a bundle with higher utility. Therefore $H$ is not strategyproof.
\end{proof}

%% file: maxmin_dse.tex
\subsection{Maxmin welfare is DSE-implementable}\label{sec:maxmin}

We will claim that Mechanism~\ref{mech:maxmin-dse} DSE-implements maxmin welfare. We are using $u_i$ as a variable in the convex program in Step 2, but we will reserve $u_i(x_i)$ for denoting agent $i$'s utility for a bundle $x_i$. We could have used $u_i \ge \gamma$ and $u_i w_{ij} \le x_{ij}$ for our first two constraints, but requiring $u_i = \gamma$ and $u_i w_{ij} = x_{ij}$ ensures a unique solution (and does not affect the optimal value).

\begin{mechanism}
\begin{enumerate}
\item Ask each agent $i$ report her set of desired goods $R_i$ (which fully specifies her utility function $u_i$ and her weights $w_{i1}\dots w_{im}$).
\item Let $(\xs, \mathbf{u^*})$ be an optimal solution the following convex program:
\begin{alignat*}{2}
\max\limits_{\substack{\x \in \bbr_{\ge 0}^{n\times m},\\ \bfu = (u_1...u_n) \in \bbrpos^n}} &\ \gamma  \\\nonumber
s.t.\ &\ u_i = \gamma &&\ \forall i \in N\\
&\ u_i w_{ij} = x_{ij} \quad\quad &&\ \forall i \in N, j \in M\\
&\ \sum\limits_{i \in N} x_{ij}\leq s_j\quad\quad &&\ \forall j \in M 
\end{alignat*}
\item Return the allocation $\xs$.
\end{enumerate}
\caption{A simple revelation mechanism which DSE-implements maxmin welfare.}
\label{mech:maxmin-dse}
\end{mechanism}

We say that an allocation $\x$ is \emph{maxmin-optimal} if the minimum utility in $\x$ is the largest possible minimum utility among all valid allocations. Formally,  $\min_{i \in N} u_i(x_i) = \max_{\xprime} \min_{i \in N} u_i(x'_i)$. 

\begin{lemma}\label{lem:maxmin-correct}
Assume agents truthfully report their desired sets of goods. Then Mechanism~\ref{mech:maxmin-dse} correctly computes a maxmin-optimal allocation.
\end{lemma}

\begin{proof}
Let $\x$ be the allocation returned by Mechanism~\ref{mech:maxmin-dse}, and suppose there exists $\y$ such that $\min_{k \in N} u_k(y_k) > \min_{k \in N} u_k(x_k)$. Consider the allocation $\xprime$ where $x'_{ij} = w_{ij} \min_{k \in N} u_k(y_k)$ for all $i,j$. Then $u_i(x'_i) =  \min_{k \in N} u_k(y_k)$ for all $i \in N$. Furthermore, $y_{ij} \ge w_{ij} u_i(y_i)$ by definition of $u_i(y_i)$, and $u_i(y_i) \ge u_i(x'_i)$, so
\[
y_{ij} \ge w_{ij} u_i(y_i) \ge w_{ij} u_i(x'_i) = x'_{ij}
\]
Thus since $\y$ is a valid allocation, so is $\xprime$. Therefore $\xprime$ is feasible for our convex program, and $\min_{i \in N} u_i(x'_i) = \min_{i \in N} u_i(y_i) > \min_{i \in N} u_i(x_i)$, so $\x$ could not have been an optimal solution to our convex program.
\end{proof}

\begin{lemma}\label{lem:maxmin-sp}
Mechanism~\ref{mech:maxmin-dse} is strategyproof.
\end{lemma}

\begin{proof}
Let $R_1\dots R_n$ be the true desired sets of goods. Suppose for sake of contradiction that there exists an instance where an agent $i$ can increase her utility by reporting some $R_i' \ne R_i$. Let $x_i$ and $x'_i$ be agent $i$'s bundles when she reports $R_i$ and $R'_i$, respectively (assuming other agents make the same reports in both cases). Due to the constraint $u_i w_{ij} = x_{ij}$, our mechanism will set $x'_{ij} = 0$ for all $j \not \in R_i'$. If $R_i' \subsetneq R_i$, then there exists a $j \in R_i$ where $j \not \in R_i'$. Thus $x'_{ij} = 0$, which implies that $u_i(x'_i)= 0$, since $j \in R_i$. 

Suppose $R_i \subsetneq R'_i$, and let $w'_{i1}\dots w'_{im}$ be the weights associated with $R'_i$. In this case, there exists a $j \in R'_i \setminus R_i$, so $w'_{ij} = 1$ and $w_{ij} = 0$. We claim that any utility vector $u_1\dots u_n$ that is feasible in the original program (when agent $i$ reports $R_i$) is also feasible in the new program (when agent $i$ reports $R'_i$). Let $x_{kj} = u_k w_{kj}$ and $x'_{kj} = u_k w'_{kj}$. Since $w'_{kj} \ge w_{kj}$ for all $k,j$, we have $x'_{kj} \ge x_{kj}$ for all $j \in M$, so $\sum_{k \in N} x_{kj} \le \sum_{k \in N} x'_{kj} \le s_j$. Thus if $\xprime$ and $u_1\dots u_n$ are feasible together, so are $\x$ and $u_1\dots u_n$. This means that the optimal value of the new program is at most the optimal value of the original program: the objective functions are the same, and the feasible set for the new program is a subset of that of the original program. Since each agent's utility is equal to the objective value of the convex program, this means that agent $i$'s utility when she reports $R'_i$ cannot improve.

Thus we have shown that reporting $R'_i \ne R_i$ cannot improve agent $i$'s utility. We conclude that this mechanism is strategyproof.
\end{proof}

\begin{theorem}\label{thm:maxmin}
Mechanism~\ref{mech:maxmin-dse} DSE-implements maxmin welfare.
\end{theorem}

\begin{proof}

Lemma~\ref{lem:maxmin-sp} implies that there is at least one DSE: in particular, truthful revelation is a DSE. Thus it remains only to show that there are no ``bad" dominant strategy equilibria, i.e., every DSE results in a maxmin-optimal allocation.

We claim that in any DSE, the vector of utilities is the same as in the truthful DSE, which we know has optimal maxmin welfare by Lemma~\ref{lem:maxmin-correct}. Since this is a revelation mechanism, each agent just reports a utility function $u_i$. Let $\bfu = u_1\dots u_n$ be the true utility profile, and let $\mathbf{u'} = u'_1\dots u'_n$ be an arbitrary DSE. Since the mechanism is strategyproof, we know that $\bfu$ is also a DSE. Thus for every agent $i$, $u_i$ and $u_i'$ are both dominant strategies (it is possible that $u_i = u'_i$). 

For each $r \in \{1\dots n+1\}$, define another utility profile $\bfu^r$ where each agent $i \in \{1\dots r-1\}$ reports $u_i$ and each agent $i \in \{r\dots n\}$ reports $u_i'$. Suppose every agent is reporting according to $\bfu^r$. If agent $r$ switches from reporting $u'_r$ to truthfully reporting $u_r$, she alters the utility profile from $\bfu^r$ to $\bfu^{r+1}$. Let $\x^r$ and $\x^{r+1}$ be the resulting allocations for reported utility profiles $\bfu^r$ and $\bfu^{r+1}$, respectively. Since reporting $u_r$ and reporting $u_r'$ are both dominant strategies for agent $r$, she must be indifferent between $\x^{r}$ and $\x^{r+1}$ (according to her true utility function, $u_r$). Formally, $u_r(x_r^r) = u_r(x_r^{r+1})$.

Next, by the definition of the mechanism, each agent has the same utility for her resulting bundle (according to the utility function she reports). Let $\gamma^r$ be every agent's utility for the allocation $\x^r$, according to her reported utility function $u_i^r$: $u_i^r(x_i^r) = \gamma^r$. Our next claim is that $u_i(x_i^r) = \gamma^r$, i.e., each agent's true utility for $\x^r$ is $\gamma^r$. As before, we know that $R_i \subseteq R'_i$ for each agent $i$: reporting $R'_i \subsetneq R_i$ always results in getting zero utility. This means that $R_i \subseteq R_i^r$. Furthermore, the convex program ensures that each agent $i$ receives the same amount of every good in her reported set $R_i^r$. Thus we have
\[
u_i^r(x_i^r) = \min_{j \in R_i^r} x_{ij}^r = \min_{j \in R_i} x_{ij}^r = u_i(x_i^r)
\]
Therefore $u_i(x_i^r) = \gamma^r$, i.e., each agent's true utility for $\x^r$ is $\gamma^r$. In particular, $u_r(x_r^r) = \gamma^r$ and $u_r(x_r^{r+1}) = \gamma^{r+1}$.

We showed above that $u_r(x_r^r) = u_r(x_r^{r+1})$, so we now have $\gamma^{r+1} = \gamma^r$ for all $r$. This implies $\gamma^1 = \gamma^{n+1}$.  Thus each agent's true utility for $\x^1$ (which is $\gamma^1$) is the same as each agent's true utility for $x^{n+1}$ (which is $\gamma^{n+1}$). By definition, $\x^{n+1}$ is the resulting allocation when each agent truthfully reports $u_i$, and $\x^1$ is the resulting allocation when each agent reports $u_i'$. Thus we have shown that each agent's utility is the same in these two allocations.

Since $\mathbf{u'}$ was an arbitrary DSE, we have shown that in any DSE, every agent's utility is the same as in the truthful outcome. Therefore the outcome of any DSE is a maxmin-optimal allocation.
\end{proof}

\subsubsection{The above mechanism does not Nash-implement maxmin welfare}\label{sec:maxmin-not-nash}

In this section, we show that our DSE implementation of maxmin welfare is not a Nash implementation, i.e., there may be Nash equilibria that are not optimal. Consider an instance with $n$ agents and $n$ goods, where each agent $i$'s true set of desired goods is $R_i = \{i\}$. Assume each good has supply 1. The unique maxmin-optimal allocation has $x_{ii} = 1$ for all $i \in N$ and $x_{ij} = 0$ for $j \ne i$, i.e., it gives the entirety of each good to the unique agent who desires it. This results in each agent having utility 1, and thus maxmin welfare of 1.

Now consider the strategy profile where each agent $i$ reports that she desires every good, i.e., reports $M$. The resulting allocation will give each agent exactly $1/n$ of each good, resulting in each agent's utility (according to her true utility function) being $1/n$. We claim that this is a Nash equilibrium. Suppose an agent $i$ reports $R'_i$ instead of $M$. If $R'_i = \emptyset$, agent $i$ receives nothing, so that cannot increase her utility. Thus let $j$ be any good in $R'_i$. Since the other $n-1$ agents are also reporting that they desire good $j$, our mechanism would divide $j$ evenly across all the agents, resulting in agent $i$ receiving $x_{ij} = 1/n$. Since our mechanism gives each agent equal the same quantity of each good in their reported set, agent $i$ does not receive more than $1/n$ of any good, so her utility is at most $1/n$ (in fact, it will be exactly $1/n$). Thus agent $i$ cannot improve her utility by bidding some $R'_i \ne M$, so the strategy profile where each agent reports $M$ is a Nash equilibrium. Furthermore, the maxmin welfare is $1/n$, which is actually a factor of $n$ worse than the optimal maxmin welfare of 1.

%% file: maxmin_nash.tex
\subsection{Maxmin welfare is Nash-implementable}\label{sec:maxmin-nash}

Our Nash implementation of maxmin welfare will use some of the same intuition from our DSE implementation. However, we will have to be careful to avoid bad equilibria like the one described in Section~\ref{sec:maxmin-not-nash}. Let $H$ denote Mechanism~\ref{mech:maxmin-dse} (which DSE-implements maxmin welfare), and let $H_X(\bfu)$ denote the allocation produced by $H$ when the reported utility profile is $\bfu$.

Mechanism~\ref{mech:maxmin-nash} operates as follows. First, it asks each agent to report not only her own utility function, but the utility function of \emph{every} agent. Recall that the utility function is specified by the set of desired goods; hence, agent $i$ reports $R_1(i)\dots R_n(i)$, where $R_k(i)$ is the set of goods that agent $i$ says agent $k$ desires.

Step 2 defines $\eta_i$ and $\barn$, which will be used later to penalize agents in a way that aligns incentives. The scalar $\eta_i$ denotes the number of agents $k$ where what agent $i$ says that agent $k$ wants ($R_k(i)$) conflicts what agent $k$ says that she wants ($R_k(k)$). We will penalize agents for having large values of $\eta_i$ to incentivize them to come to a consensus. The set $\barn$ contains the set of agents $i$ where for some other agent $k$, agent $i$ is saying that agent $k$ wants more goods than agent $k$ is saying that she actually wants (i.e., $R_k(k) \subsetneq R_k(i)$). We will penalize this specific type of disagreement more strongly; the reason will become clear in the proof of Theorem~\ref{thm:maxmin-nash}. 

Step 3 defines $\alpha_i \in [0,1]$: $\alpha_i = 1$ represents no penalty, and $\alpha_i = 0$ represents an absolute penalty (i.e., that agent will end up with no utility). Specifically, $\alpha_i = 0$ for each agent in $\barn$, and for those agents not in $\barn$, a higher $\eta_i$ leads to a higher penalty. Next, we use the mechanism $H$ to compute a maxmin-optimal allocation $\xprime$ \emph{ignoring agents in $\barn$}. It is crucial that we ignore those agents when computing this allocation. Finally, we return an allocation $\x$ which is just $\xprime$ with the $\alpha_i$ penalties applied. As usual, $x_i$ and $y_i$ are vectors in $\bbrpos^m$.

The first thing to notice is that we have solved the problem from Section~\ref{sec:maxmin-not-nash}. We claim that each agent $i$ reporting $R_k(i) = M$ for all $k \in N$ is not longer a Nash equilibrium when agent $i$'s true desired set is a strict subset of $M$. This is because if agent $i$ shrinks $R_i(i)$ to her true subset, but reports $R_k(i) = M$ for all $k \ne i$, $\barn$ will contain every agent except for her. This means that she would get all of the resources in Step 3, which clearly increases her utility. The other agents would then respond by setting $R_i(k) = R_i(i)$ so that they are no longer in $\barn$, but this at least shows that $R_k(i) = M$ for all $i,k\in N$ is not a Nash equilibrium.

\begin{mechanism}
\begin{enumerate}
\item Ask each agent $i$ to report $R_1(i)\dots R_n(i)$, where $R_k(i) \subseteq M$ for each $k \in N$.
\item For each $i \in N$, let $\eta_i = |\{k \in N: R_k(k) \ne R_k(i) \}|$. Define the set $\barn$ by $\barn = \{i \in N: \exists k \in N \text{ s.t. } R_k(k) \subsetneq R_k(i)\}$.
\item For each $i \in N$, define $\alpha_i \in [0,1]$ as follows. If $i \in \barn$, then $\alpha_i = 0$. Otherwise, let $\alpha_i = 1-\eta_i/n$.
\item Let $\bfu$ be the utility profile where the set of goods desired by agent $i$ is $R_i(i)$ if $i\not \in \barn$, and is $\emptyset$ if $i \in \barn$. Let $\y = H_X(\bfu)$.
\item Return the allocation $\x$ where for each $i \in N$, $x_i = \alpha_i y_i$.
\end{enumerate}
\caption{A mechanism which Nash-implements maxmin welfare.}
\label{mech:maxmin-nash}
\end{mechanism}

For the rest of this section, we will use $\mathbf{\tilr} = \tilr_1\dots \tilr_n$ to denote the true desired sets of goods, and $\mathbf{\tilde{u}} = \tilde{u_1}\dots \tilde{u_n}$ to denote the corresponding utility profile.

\begin{lemma}\label{lem:maxmin-nash-truthful}
When each agent reports $\mathbf{\tilr}$, Mechanism~\ref{mech:maxmin-nash} returns a maxmin-optimal allocation.
\end{lemma}

\begin{proof}
In this case, we have $\barn = \emptyset$ and $\alpha_i = 1$ for all $i \in N$. Thus the $\bfu$ used in Step 4 is the true utility profile, so $H$ computes a maxmin-optimal $\y$ allocation with respect to the true preferences. Since $\alpha_i = 1$ for all $i \in N$, we have $\x = \y$, so Mechanism~\ref{mech:maxmin-nash} does indeed return a maxmin-optimal allocation.
\end{proof}

\begin{lemma}\label{lem:maxmin-truthful-nash-eq}
The strategy profile where each agent reports $\mathbf{\tilr}$ is a Nash equilibrium.
\end{lemma}

\begin{proof}
Suppose the opposite: then there exists an agent $i$ who can report $R'_1(i) \dots R'_n(i)$ and increase her utility. When all agents report $\mathbf{\tilr}$, let $\tilde{\alpha}_i$ be agent $i$'s value of $\alpha_i$, let $\mathbf{\tilde{y}}$ be the intermediate allocation produced in Step 4, and let $\mathbf{\tilde{x}}$ be the final resulting allocation. When agent $i$ reports $R'_1(i) \dots R'_n(i)$ and all other agents report $\mathbf{\tilr}$, we use $\alpha'_i$, $\mathbf{y'}$, and $\xprime$ analogously. 

Thus we have assumed that $\tilde{u}_i(x'_i) > \tilde{u}_i (\tilde{x}_i)$, i.e., she is strictly happier when she deviates and reports $R'_1(i) \dots R'_n(i)$. Since $x'_i = \alpha'_i y'_i$ and $\tilde{x}_i = \tilde{\alpha}_i \tilde{y}_i$, we have $\tilde{u}_i(\alpha'_i y'_i) > \tilde{u}_i (\tilde{\alpha}_i \tilde{y}_i)$. When all agents report $\mathbf{\tilr}$, all agents are in agreement, so $\tilde{\alpha}_i = 1$. Since $\alpha'_i \le 1$, we have $\alpha'_i \le \tilde{\alpha}_i$. Therefore we must have $\tilde{u}_i(y'_i) > \tilde{u}_i (\tilde{y}_i)$.

We claim that after this deviation, $\barn = \emptyset$. If $i \in \barn$, she receives zero utility, so such a deviation could not help her. Since all agents $k\ne i$ report the same thing, the only way for agent $k\ne i$ to be in $\barn$ is if what agent $k$ is reporting that agent $i$ wants (which in this case is $\tilr_i$) is a superset of what agent $i$ is reporting that she wants (which in this case is $R'_i(i)$). If $R'_i(i) \subsetneq \tilr_i$, then there is a good $j \in \tilr_i$ that agent $i$ will receive none of, i.e., $y'_{ij} = 0$. This is because the convex program in $H$ will only allocate agent $i$ a portion of good $j$ if good $j$ is in her reported set. Thus $R'_i(i) \subsetneq \tilr_i$ implies that $\tilde{u}_i(y'_i) = 0$. Therefore agent $i$'s utility cannot have improved, which is a contradiction.

Therefore after the deviation, $\barn = \emptyset$. This means that $\mathbf{y'}$ is just the maxmin-optimal allocation computed by Mechanism~\ref{mech:maxmin-dse} for utility profile $\tilr_1, \tilr_2\dots R'_i(i)\dots \tilr_n$. But this implies that agent $i$ is improving her utility for the allocation produced by Mechanism~\ref{mech:maxmin-dse} by reporting $R'_i(i)$ instead of $\tilr_i$. This contradicts the strategyproofness of Mechanism~\ref{mech:maxmin-dse} (by Lemma~\ref{lem:maxmin-sp}). Thus $\tilde{u}_i(y'_i) > \tilde{u}_i (\tilde{y}_i)$ is impossible, which implies that each agent reporting $\mathbf{\tilr}$ is in fact a Nash equilibrium.
\end{proof}

\begin{theorem}\label{thm:maxmin-nash}
Mechanism~\ref{mech:maxmin-nash} Nash-implements maxmin welfare.
\end{theorem}

\begin{proof}
We need to show that for each problem instance, Mechanism~\ref{mech:maxmin-nash} has at least one Nash equilibrium, and every Nash equilibrium is optimal. Lemma~\ref{lem:maxmin-truthful-nash-eq} implies that at least one Nash equilibrium exists, so it remains to show that every Nash equilibrium results in a maxmin-optimal allocation.

Consider an arbitrary Nash equilibrium where agent $i$ reports $R_1(i) \dots R_n(i)$. First, note that each agent $i$ can always achieve $i \not \in \barn$ and $\alpha_i = 1$ by having $R_k(i) = R_k(k)$ for each $k \ne i$. This also does not restrict what she reports for $R_i(i)$, which is what actually affects the allocation $\y$. Thus in any Nash equilibrium, we have $\barn = \emptyset$, $\alpha_i = 1$ for all $i \in N$, and $R_k(i) = R_k(k)$ for all $i,k \in N$.

Since $H$ allocates a portion of good $j$ to agent $i$ only if $j$ is in agent $i$'s reported set, if $R_i(i) \subsetneq \tilr_i$, then agent $i$ will receive zero utility. Reporting $R_i(i) = \tilr_i$ instead (and still reporting $R_k(i) = R_k(k)$ for $k \ne i$, so that she is not in $\barn$) would give her nonzero utility, so $R_i(i) \subsetneq \tilr_i$ is impossible in a Nash equilibrium. Thus we must have $\tilr_i \subseteq R_i(i)$ for all $i \in N$.

Now suppose that $\tilr_i \subsetneq R_i(i)$. Suppose that agent $i$ reports $R'_i(i) = \tilr_i$ instead (and reports the same $R_k(i)$ for each $k \ne i$). Then since every agent $k \ne i$ is reporting $R_i(k) = R_i(i)$, we now have $R'_i(i) \subsetneq R_i(k)$ for each $k \ne i$. This means that $\barn$ contains every agent except for $i$, so every agent other than $i$ is ignored when computing the allocation $\y$. This means that the allocation $\y$ gives agent $i$ her maximum possible utility (which is $\min_{j \in \tilr_i} s_j$: the minimum supply of any good she desires). Since agent $i$ is still reporting $R_k(i)= R_k(k)$ for all $k\ne i$, she still has $\alpha_i = 1$, which means that in the final allocation $\x$, she receives her maximum possible utility.

Therefore in any Nash equilibrium, for each $i \in N$, either $R_i(i) = \tilr_i$ (i.e., she is reporting her true set), or $\tilr_i \subsetneq R_i(i)$ and agent $i$ is receiving her maximum possible utility. We proceed by case analysis.

Case 1: Every agent agent is reporting $R_i(i) = \tilr_i$. Since we have $R_k(i) = R_k(k)$ for all $i,k \in N$, each agent must be reporting $\mathbf{\tilr}$. Then by Lemma~\ref{lem:maxmin-nash-truthful}, we get a maxmin-optimal allocation in this case.

Case 2: At least one agent $i$ is reporting $\tilr_i \subsetneq R_i(i)$ and thus is receiving her maximum possible utility of $\min_{j \in \tilr_i} s_j$. Let $\gamma=\min_{j \in \tilr_i} s_j$: then for any allocation $\xprime$,
\[
\min_{k \in N} u_k(x_k') \le u_i(x_i') \le \gamma
\]
i.e., the value of the maxmin objective can never be more than $\gamma$.
Since $H$ gives each agent the same utility (according to her reported preferences), for all $k \in N$ we have
\[
\gamma = \min_{j \in R_k(k)} x_{kj} \ge \min_{j \in \tilr_k} = \tilde{u}_k(x_k)
\]
where the inequality is because $\tilr_i\subseteq R_i(i)$. Thus we have $\min_{k \in N} \tilde{u}_k(x_k) = \gamma$. Thus $\x$ must be maxmin-optimal. 

Therefore we have shown that in either case, the Nash equilibrium must result in a maxmin-optimal allocation. We conclude that Mechanism~\ref{mech:maxmin-nash} Nash-implements maxmin welfare.
\end{proof}

%% file: conclusion.tex
\section{Conclusion and future work}

In this paper, we showed that every CES welfare function except $\rho = 1$ can be Nash-implemented by an augmented trading post mechanism. This strengthened previous results which only handled Nash welfare~\cite{branzei_nash_2017} or assumed agents did not behave strategically~\cite{goel_beyond_2018}. Next, we showed that DSE implementation for this problem is generally impossible, with the exception of maxmin welfare, where a simple revelation mechanism does indeed DSE-implement maxmin welfare. Although this revelation mechanism does not Nash-implement maxmin welfare, we were able to Nash-implement maxmin welfare with a different mechanism.

We were not able to resolve whether utilitarian welfare is Nash-implementable for bandwidth allocation. Our trading post mechanism breaks down in this setting, since $f_j(b) = b^{1-1} = 1$ is not a valid constraint curve. Maskin's monotonicity approach is not viable either, since utilitarian welfare does not satisfy no veto power. We leave this as an open question.

Another interesting direction would be to extend these results to a wider range of utility functions. Our reduction between price curves and trading post means that if price curve equilibria maximizing CES welfare were shown to exist for a wider range of utility functions, it seems likely that our Nash implementation results would carry over as well (depending on the form of the price curves).

It would also be interesting to consider another dimension of strategic behavior by allowing agents to choose which path in the network to use. In this case, we could write each agent's utility function as $u_i(x_i) = \max_{p \in P_i} \min_{j \in p} x_{ij}$, where $P_i$ is the set of paths from agent $i$'s desired source to desired destination. This is reminiscent of routing games, in that agents are strategically choosing their paths, but still distinct, in that each agent may use the same link in different quantities (i.e., receive different amounts of bandwidth). Although this model is less accurate in terms of how the internet actually works (see Section~\ref{sec:intro}), it may be an appropriate model for other situations.

More broadly, we feel that trading post is a powerful mechanism that is able to simulate a price-taking market while also handling strategic behavior. We wonder if trading post, or variants thereof, may be useful in designing mechanisms for other resource allocation problems as well.

%% file: main.bbl
\begin{thebibliography}{10}

\bibitem{adsul_nash_2010}
Bharat Adsul, Ch.~Sobhan Babu, Jugal Garg, Ruta Mehta, and Milind Sohoni.
\newblock Nash equilibria in fisher market.
\newblock In {\em Proceedings of the Third International Conference on
  Algorithmic Game Theory}, SAGT'10, pages 30--41, Berlin, Heidelberg, 2010.
  Springer-Verlag.

\bibitem{arrow_existence_1954}
Kenneth~J. Arrow and Gerard Debreu.
\newblock Existence of an equilibrium for a competitive economy.
\newblock {\em Econometrica}, 22(3):265--290, 1954.

\bibitem{atkinson_1970_measurement}
Anthony~B Atkinson.
\newblock On the measurement of inequality.
\newblock {\em Journal of Economic Theory}, 2(3):244 -- 263, 1970.

\bibitem{bergson_1938_reformulation}
Abram Bergson.
\newblock A reformulation of certain aspects of welfare economics.
\newblock {\em The Quarterly Journal of Economics}, 52(2):310--334, 1938.

\bibitem{blackorby_1978_measures}
Charles Blackorby and David Donaldson.
\newblock Measures of relative equality and their meaning in terms of social
  welfare.
\newblock {\em Journal of Economic Theory}, 18(1):59 -- 80, 1978.

\bibitem{brainard_how_2005}
William~C. Brainard and Herbert~E. Scarf.
\newblock How to compute equilibrium prices in 1891.
\newblock {\em American Journal of Economics and Sociology}, 64(1):57--83,
  2005.

\bibitem{branzei_fisher_2014}
Simina Br\^{a}nzei, Yiling Chen, Xiaotie Deng, Aris Filos-Ratsikas, S{\o}ren
  Kristoffer~Stiil Frederiksen, and Jie Zhang.
\newblock The fisher market game: Equilibrium and welfare.
\newblock In {\em Proceedings of the Twenty-Eighth AAAI Conference on
  Artificial Intelligence}, AAAI'14, pages 587--593. AAAI Press, 2014.

\bibitem{branzei_nash_2017}
Simina Branzei, Vasilis Gkatzelis, and Ruta Mehta.
\newblock Nash social welfare approximation for strategic agents.
\newblock In {\em Proceedings of the 2017 ACM Conference on Economics and
  Computation}, EC '17, pages 611--628, New York, NY, USA, 2017. ACM.

\bibitem{buchanan_toward_1980}
James~M Buchanan, Robert~D Tollison, and Gordon Tullock.
\newblock {\em Toward a theory of the rent-seeking society}.
\newblock Number~4. Texas A \& M Univ Pr, 1980.

\bibitem{cerf_1974_protocol}
V.~Cerf and R.~Kahn.
\newblock A protocol for packet network intercommunication.
\newblock {\em IEEE Transactions on Communications}, 22(5):637--648, May 1974.

\bibitem{cole_mechanism_2013}
Richard Cole, Vasilis Gkatzelis, and Gagan Goel.
\newblock Mechanism design for fair division: Allocating divisible items
  without payments.
\newblock In {\em EC 2013}, 2013.

\bibitem{dalton_1920_measurement}
Hugh Dalton.
\newblock The measurement of the inequality of incomes.
\newblock {\em The Economic Journal}, 30(119):348--361, 1920.

\bibitem{dasgupta_impl_1979}
Partha Dasgupta, Peter Hammond, and Eric Maskin.
\newblock The implementation of social choice rules: Some general results on
  incentive compatibility.
\newblock {\em The Review of Economic Studies}, 46(2):185--216, 1979.

\bibitem{dubey_theory_1978}
Pradeep Dubey and Martin Shubik.
\newblock A theory of money and financial institutions. 28. the non-cooperative
  equilibria of a closed trading economy with market supply and bidding
  strategies.
\newblock {\em Journal of Economic Theory}, 17(1):1 -- 20, 1978.

\bibitem{eisenberg_aggregation_1961}
E.~Eisenberg.
\newblock Aggregation of {Utility} {Functions}.
\newblock {\em Management Science}, 7(4):337--350, July 1961.

\bibitem{eisenberg_consensus_1959}
Edmund Eisenberg and David Gale.
\newblock Consensus of {Subjective} {Probabilities}: {The} {Pari}-{Mutuel}
  {Method}.
\newblock {\em The Annals of Mathematical Statistics}, 30(1):165--168, March
  1959.

\bibitem{feldman_proportional_2009}
M.~{Feldman}, K.~{Lai}, and L.~{Zhang}.
\newblock The proportional-share allocation market for computational resources.
\newblock {\em IEEE Transactions on Parallel and Distributed Systems},
  20(8):1075--1088, Aug 2009.

\bibitem{floyd_1993_random}
S.~Floyd and V.~Jacobson.
\newblock Random early detection gateways for congestion avoidance.
\newblock {\em IEEE/ACM Transactions on Networking}, 1(4):397--413, Aug 1993.

\bibitem{giraud_strategic_2003}
Ga{\"e}l Giraud.
\newblock Strategic market games: an introduction.
\newblock {\em Journal of Mathematical Economics}, 39(5):355 -- 375, 2003.
\newblock Strategic Market Games.

\bibitem{goel_beyond_2018}
Ashish Goel, Reyna Hulett, and Benjamin Plaut.
\newblock Markets beyond nash welfare for leontief utilities.
\newblock {\em CoRR}, abs/1807.05293, 2018.

\bibitem{jain_eisenberggale_2010}
Kamal Jain and Vijay~V. Vazirani.
\newblock Eisenberg–{Gale} markets: {Algorithms} and game-theoretic
  properties.
\newblock {\em Games and Economic Behavior}, 70(1):84--106, September 2010.

\bibitem{kaneko_nash_1979}
Mamoru Kaneko and Kenjiro Nakamura.
\newblock The nash social welfare function.
\newblock {\em Econometrica}, 47(2):423--35, 1979.

\bibitem{kelly_1998_rate}
F~P Kelly, A~K Maulloo, and D~K~H Tan.
\newblock Rate control for communication networks: shadow prices, proportional
  fairness and stability.
\newblock {\em Journal of the Operational Research Society}, 49(3):237--252,
  Mar 1998.

\bibitem{papa_worst_1999}
Elias Koutsoupias and Christos Papadimitriou.
\newblock Worst-case equilibria.
\newblock In {\em Proceedings of the 16th Annual Conference on Theoretical
  Aspects of Computer Science}, STACS'99, pages 404--413, Berlin, Heidelberg,
  1999. Springer-Verlag.

\bibitem{maskin_nash_1999}
Eric Maskin.
\newblock Nash equilibrium and welfare optimality.
\newblock {\em The Review of Economic Studies}, 66(1):23--38, 1999.

\bibitem{maskin_2002_book}
Eric Maskin and T.~Sj{\"o}str{\"o}m.
\newblock {\em Implementation Theory}, pages 237--288.
\newblock North Holland, Amsterdam, 2002.

\bibitem{matros_chinese_2011}
Alexander Matros.
\newblock Chinese auctions.
\newblock In {\em GAME THEORY AND MANAGEMENT. Collected abstracts of papers
  presented on the Fifth International Conference Game Theory and
  Management/Editors Leon A. Petrosyan and Nikolay A. Zenkevich.--SPb.:
  Graduate School of Management SPbU, 2011.--268 p. The collection contains
  abstracts of papers accepted for the Fifth International}, page 153, 2011.

\bibitem{moulin_2003_fair}
Herv{\'e} Moulin.
\newblock {\em Fair Division and Collective Welfare}, chapter~3.
\newblock MIT Press, January 2003.

\bibitem{nash_bargaining_1950}
John Nash.
\newblock {The Bargaining Problem}.
\newblock {\em Econometrica}, 18(2):155--162, April 1950.

\bibitem{nisan_2007_algorithmic}
Noam Nisan, Tim Roughgarden, {\'E}va Tardos, and Vijay~V Vazirani, editors.
\newblock {\em Algorithmic Game Theory}.
\newblock Cambridge University Press, 2007.

\bibitem{padhye_1998_modeling}
Jitendra Padhye, Victor Firoiu, Don Towsley, and Jim Kurose.
\newblock Modeling tcp throughput: A simple model and its empirical validation.
\newblock {\em SIGCOMM Comput. Commun. Rev.}, 28(4):303--314, October 1998.

\bibitem{pan_2000_choke}
Rong Pan, B.~Prabhakar, and K.~Psounis.
\newblock Choke - a stateless active queue management scheme for approximating
  fair bandwidth allocation.
\newblock In {\em Proceedings IEEE INFOCOM 2000. Conference on Computer
  Communications. Nineteenth Annual Joint Conference of the IEEE Computer and
  Communications Societies (Cat. No.00CH37064)}, volume~2, pages 942--951
  vol.2, 2000.

\bibitem{pigou_1912_wealth}
A.C. Pigou.
\newblock {\em Wealth and Welfare}.
\newblock PCMI collection. Macmillan and Company, limited, 1912.

\bibitem{rawls_1971_theory}
John Rawls.
\newblock {\em A Theory of Justice}.
\newblock Belknap Press of Harvard University Press, Cambridge, Massachussets,
  1 edition, 1971.

\bibitem{roughgarden_selfish_2005}
Tim Roughgarden.
\newblock {\em Selfish Routing and the Price of Anarchy}.
\newblock The MIT Press, 2005.

\bibitem{samuelson_1947_foundations}
Paul~Anthony Samuelson.
\newblock Foundations of economic analysis.
\newblock 1947.

\bibitem{sen_1976_welfare}
Amartya Sen.
\newblock Welfare inequalities and rawlsian axiomatics.
\newblock {\em Theory and Decision}, 7(4):243--262, Oct 1976.

\bibitem{sen_1977_social}
Amartya Sen.
\newblock Social choice theory: A re-examination.
\newblock {\em Econometrica}, 45(1):53--89, 1977.

\bibitem{shapley_trade_1977}
Lloyd Shapley and Martin Shubik.
\newblock Trade {Using} {One} {Commodity} as a {Means} of {Payment}.
\newblock {\em Journal of Political Economy}, 85(5):937--68, 1977.

\bibitem{varian_equity_1974}
Hal Varian.
\newblock Equity, envy, and efficiency.
\newblock {\em Journal of Economic Theory}, 9(1):63 -- 91, 1974.

\bibitem{vazirani_2007_combinatorial}
Vijay~V Vazirani.
\newblock Combinatorial algorithms for market equilibria.
\newblock {\em Algorithmic Game Theory}, pages 103--134, 2007.

\bibitem{walras_elements_1874}
L.~Walras.
\newblock {\em {\'E}l{\'e}ments d'{\'e}conomie politique pure; ou, Th{\'e}orie
  de la richesse sociale}.
\newblock Number v. 1-2. Corbaz, 1874.

\end{thebibliography}
